\documentclass[12pt]{article}%
\usepackage{amssymb,fullpage}
\usepackage{amsfonts}
\usepackage{amsmath}
\usepackage{graphicx,enumerate}

 \numberwithin{equation}{section}
\usepackage{amsfonts, amsthm, thmtools,graphicx,latexsym,url,epsf, algpseudocode, verbatim, setspace, float, paralist}

\usepackage[usenames,dvipsnames,svgnames,table]{xcolor}
\usepackage[pagebackref,letterpaper=true,colorlinks=true,pdfpagemode=none,citecolor=OliveGreen,linkcolor=BrickRed,urlcolor=BrickRed,pdfstartview=FitH]{hyperref}

\usepackage{tikz}

\newcommand{\hi}

\newcommand{\bmath}{\begin{equation}}
\newcommand{\emath}{\end{equation}}
\newcommand{\bmathnn}{\begin{eqnarray*}}
\newcommand{\emathnn}{\end{eqnarray*}}

\DeclareMathOperator*{\argmax}{arg\,max}
\DeclareMathOperator*{\argmin}{arg\,min}
\declaretheorem[numberwithin=section]{theorem}
\declaretheorem[sibling=theorem]{lemma}
\declaretheorem[sibling=theorem]{proposition}
\declaretheorem[sibling=theorem]{claim}

\declaretheorem[sibling=theorem]{corollary}

\declaretheorem[sibling=theorem]{definition}

\def\N{\ensuremath{\mathbb{N}}}

\def\R{\mathbb{R}}
\def\alg{\textup{ALG}}
\def\poolsize{\zeta}
\def\eps{\epsilon}
\def\optoff{\textup{OMN}}
\def\opton{\textup{OPT}}
\def\algGreedy{{\textup{Greedy}}}
\def\algcritical{{\textup{Patient}}}

\def\A{A}
\def\a{a}
\def\setminus{-}
\def\1{{\bf{1}}}
\def\b{b}

\renewcommand{\P}[1]{{\mathbb{P}}\left[#1\right]}
\newcommand{\PP}[2]{{\mathbb{P}}_{#1}\left[#2\right]}
\newcommand{\E}[1]{{\mathbb{E}}\left[#1\right]}
\newcommand{\EE}[2]{{\mathbb{E}}_{#1}\left[#2\right]}
\newcommand{\norm}[1]{\|#1\|}
\DeclareMathOperator{\unif}{unif}

\DeclareMathOperator{\mix}{mix}

\DeclareMathOperator{\TV}{TV}
\DeclareMathOperator{\polylog}{polylog}

\newcommand\loss[1]{\mathbf{L}(#1)}
\newcommand\W[1]{\mathbf{W}(#1)}
\newcommand\apatient[1]{\textup{Patient}(#1)}
\newcommand\amechpatient[1]{\textup{Patient-Mechanism}(#1)}

\newenvironment{proofof}[1]{{\medbreak\noindent \em Proof of #1.  }}{\hfill\qed\medbreak}
\newenvironment{proofoverview}{{\medbreak\noindent\em Proof Overview. }}{\hfill\qed\medbreak}
\usepackage{color}

\def\C{\ensuremath{\mathcal{C}}}

\begin{document}

\title{Dynamic Matching Market Design\thanks{We thank Paul Milgrom and Alvin Roth for valuable comments and suggestions. We also thank Itai Ashlagi, Timothy Bresnahan,
Gabriel Carroll, Fuhito Kojima, Matthew Jackson, Muriel Niederle, Afshin Nikzad, Malwina Luczak, Michael Ostrovsky,  Bob Wilson, and Alex Wolitzky for their valuable comments, as well as several seminar participants for helpful suggestions. All errors remain our own.} }
\date{Draft: February 2014}

\author{Mohammad Akbarpour\thanks{Department of Economics, Stanford University. Email: \protect\url{mohamwad@stanford.edu}} \and Shengwu Li\thanks{Department of Economics, Stanford University. Email: \protect\url{shengwu@stanford.edu}} \and Shayan Oveis Gharan\thanks{Computer Science Division, U.C. Berkeley. Supported by a Miller fellowship. Email: \protect\url{oveisgharan@berkeley.edu}}}

\maketitle



\begin{abstract}


We introduce a simple benchmark model of dynamic matching in networked markets, where agents arrive and depart stochastically and the network of acceptable transactions among agents forms a random graph.
We analyze our model from three perspectives: waiting, optimization, and information. The main insight of our analysis is that waiting to thicken the market can be substantially more important than increasing the speed of transactions, and this is quite robust to the presence of waiting costs. From an optimization perspective, na\"{i}ve local algorithms, that choose the right time to match agents but do not exploit global network structure, can perform very close to optimal algorithms. From an information perspective, algorithms that employ even partial information on agents'  departure times perform substantially better than those that lack such information. To elicit agents' departure times, we design an incentive-compatible continuous-time dynamic mechanism without transfers.

\end{abstract}

\strut


\strut


\thispagestyle{empty}
\newpage
\thispagestyle{empty}
\tableofcontents
\newpage
\onehalfspacing

\setcounter{page}{1}
\section{Introduction}
Economics has extensively studied the problem of matching in static settings (see \cite{GS62,CK81,KC82,RS92,Roth04,HM05,Roth07,HK10}).  In such settings, a social planner observes the set of agents, and their preferences over partners, or contracts, or sets thereof.  The planner's problem is to find a matching algorithm with desirable properties; \emph{e.g.} stability, efficiency, or strategy-proofness.  The algorithm is run, the matching is made, and the world ends. Some seasonal markets, such as school choice systems and the National Residencies Matching Program, are appropriately conceived as static assignment problems without inter-temporal spillovers.

There has been comparatively little study of dynamic matching, \emph{viz.} the problem of finding good algorithms for matching markets that evolve over time, where the planner must make allocations in the present while facing an uncertain future.  In such markets, agents arrive at different times, with needs that may be fulfilled by an appropriate transaction brokered by a social planner. If their needs are not met, they may depart or their needs may expire, removing the possibility of a future transaction.  Stochastic arrivals and departures are an institutional feature of many real world marketplaces. Some natural examples are:

\begin{itemize}
\item \textbf{Kidney exchange:} In paired kidney exchanges, patient-donor pairs arrive gradually over time. They stay in the market to find a compatible pair, and may leave the market if the patient's condition deteriorates to the point where kidney transplants become infeasible.
\item \textbf{Real estate markets:} In real estate markets, agents (buyers and sellers) arrive at the market non-simultaneously, trying to rent/buy a property with certain characteristics, or to find a renter with certain characteristics.  The central broker (say, a real estate agent) has substantial information about buyers and sellers' characteristics.  Potential landlords and potential renters may leave the market if, for example, they find a gainful trade elsewhere.
\item \textbf{Allocation of workers to time-sensitive tasks}:  Both within
firms, and in online labor markets such as oDesk, planners have to
allocate workers to tasks that are profitable to undertake.  Tasks
arrive continuously, and different workers are suited to different
tasks.  Tasks may expire, and workers may cease to be available.
\end{itemize}

In static matching markets, such as the canonical stable marriage problem introduced by Gale and Shapley \cite{GS62}, the social planner solves a one-time problem by choosing which agents to match.  In contrast, in a dynamic setting, the passage of time brings new possibilities for transactions and extinguishes existing possibilities.  For example, as shown in \autoref{fig:stablematching}, the planner can match $s_1$ to $b_1$ at time $t_0$, but then he cannot match $s_2$ and $b_2$ in the next period. In contrast, if he waits and postpones matching agents for one period, he can match $s_1$ to $b_2$ and $s_2$ to $b_1$, which is a better solution.
Hence, in a dynamic setting, the social planner confronts a new question: ``\emph{When} should agents be matched?'' The first goal of this paper is to answer this question.
One might think that this is a second-order question, in the sense that, as in the static setting, the first-order question for a social planner is \emph{which} agents to match. The second goal of this paper is to compare the relative importance of these two questions.
This helps us to think systematically about the trade-offs policymakers face in settings where the timing of transactions is a design parameter, rather than an exogenous constraint.

\begin{figure}
\centering
\begin{tikzpicture}
\begin{scope}[shift={(-3.5,0)}]
\node at (-1.5,0) (a) {\includegraphics[scale=0.08]{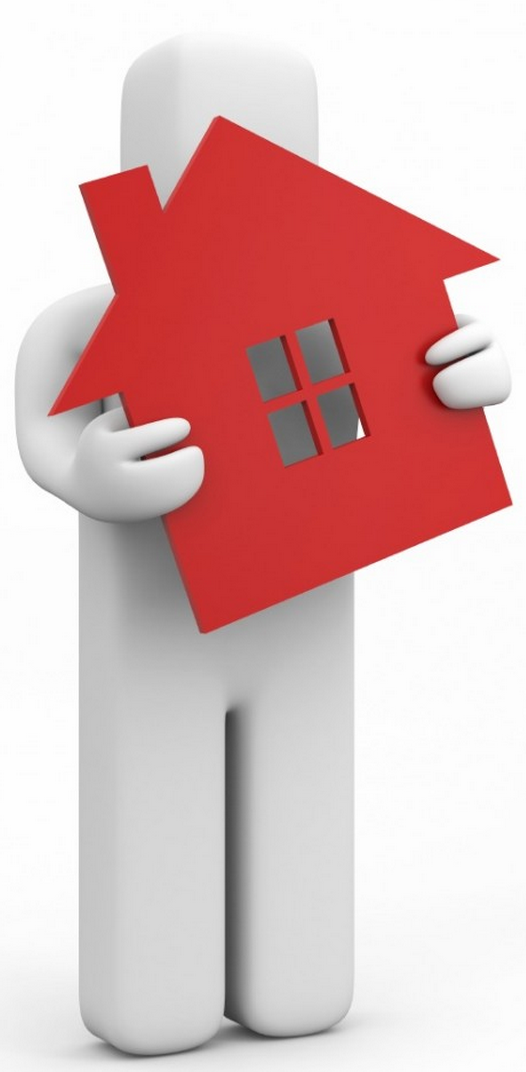}};
\node [anchor=north east, inner sep=10] at (a) {$s_1$};
\node at (1.5,0) (b) {\includegraphics[scale=0.12]{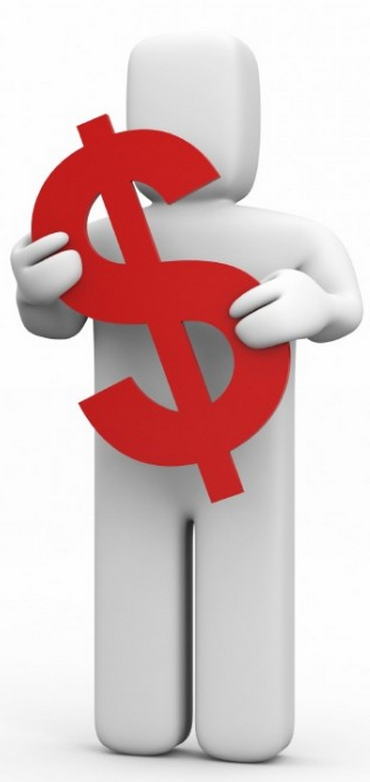}} edge (a);
\node [anchor=north west,inner sep=10] at (b) {$b_1$};
\node [anchor=south] at (0,-3) {Market at time $t_0$};
\end{scope}
\begin{scope}[shift={(3.5,0)}]
\node at (-0.5, 1) (a) {\includegraphics[scale=0.08]{seller.png}};
\node at (2.5, 1) (b) {\includegraphics[scale=0.12]{buyer.png}} edge (a);
\node at (-0.5, -1) (c) {\includegraphics[scale=0.08]{seller.png}} edge (b);
\node at (2.5, -1) (d) {\includegraphics[scale=0.12]{buyer.png}} edge (a);
\node [anchor=north east,inner sep=10] at (a) {$s_1$};
\node [anchor=north west,inner sep=10] at (b) {$b_1$};
\node [anchor=north east, inner sep=10] at (c) {$s_2$};
\node [anchor=north west, inner sep=10] at (d) {$b_2$};
\node [anchor=south] at (1,-3) {Market at time $t_0+1$};
\end{scope}
\end{tikzpicture}
\caption{A schematic real estate market, where there is a line between a seller and a buyer if and only if they are mutually interested in trade. In a static setting at time $t_0$, the optimum solution matches $s_1$ to $b_1$. But in a dynamic setting a new seller and a new buyer may arrive at the market in the next period, so the planner can wait and obtain a better solution by matching $s_1$ to $b_2$ and $s_2$ to $b_1$.}
\label{fig:stablematching}
\end{figure}

In this paper, we analyze a stylized but illuminating model of dynamic matching on networks. In our model, agents arrive at the market according to Poisson processes, and remain in the market for an interval drawn from an exponential distribution.  Abstracting from prices, our model uses simple preferences, where transactions between agents are either acceptable or unacceptable, generated according to a known distribution.  The network of acceptable transactions forms (under certain conditions) an Erd\H{o}s-R\'{e}nyi random graph.  Agents do not observe the set of acceptable transactions, and are reliant upon the planner to match them to each other.  We say that an agent \emph{perishes} if she leaves the market unmatched.  

The planner's problem is to design a matching algorithm; that is, at any point in time, to select a subset of acceptable transactions and broker those trades.  The planner observes the current set of agents and acceptable transactions, but has only probabilistic knowledge about the future.  The planner may have knowledge about which agents' needs are urgent, in the sense that he may know which agents will perish imminently if not matched.  
The goal of the planner  is to maximize the sum of the discounted utilities of all agents. In the important special case where the cost of waiting is zero, the planner's goal is equivalent to minimizing the proportion of agents who perish.   


We design two simple matching algorithms and a class of interpolating algorithms. The \emph{Greedy Algorithm} attempts to match each agent upon her arrival; it treats each instant as a static matching problem without regard for the future. The \emph{Patient Algorithm} attempts to match only urgent cases. The \emph{$Patient(\alpha)$ Algorithm}, is a class of variants of the Patient algorithm with a smaller average waiting time determined by a parameter $\alpha \geq 0$. This class is specifically appealing when dealing with discounting.  All these algorithms are \emph{local}, in the sense that they look only at the immediate neighbors of agents rather than at the global graph structure.  We compare the performance of these algorithms with two optimal benchmarks, one that has \emph{unlimited computational power} and one that has \emph{perfect foresight}.


We first consider the case of perfectly patient agents, and then show that our results are not substantially altered by discounting. Our results are as follows:  First, we show that waiting to thicken the market can dramatically reduce the rate of perishing compared to `optimal static' policies.  In particular, the fraction of perished agents under the Patient algorithm is exponentially smaller (in the average degree of agents) than the fraction under the Greedy algorithm.  This quantifies the value of waiting in dynamic matching markets.  For instance, consider a market where 1,000  new agents arrive each period, each of them stays for one period on average, and the probability of an acceptable transaction between two agents is $\frac{1}{50}$. Suppose that the market runs for 1,000 periods. Then the expected number of perished agents under the Greedy algorithm is at least 23,800, whereas the same number for the Patient algorithm is no more than $23$.


Our second result is that gains from detailed consideration of the graph structure are small compared to the gains from good timing decisions.  That is, we show that local algorithms that do not exploit the entire graph structure are close to the optimal benchmarks.  If the planner cannot identify urgent cases, then the Greedy algorithm is close to optimal.  Correspondingly, if the planner can identify urgent cases, then the Patient algorithm is close to optimal.  This is very fortunate, because the complexity of the graph structure makes it intractable to compute optimal solutions via standard dynamic programming techniques.  A corollary of this result is that urgency information and waiting are complementary:  It is only when the planner is informed about urgent cases that waiting can produce large gains.



Our third result is that our welfare comparisons are quite robust to the presence of waiting costs. In many real-world marketplaces, brokers employ variants of the Greedy Algorithm to minimize the waiting time of the agents.\footnote{For instance, our analysis of the Greedy Algorithm encompasses waiting list policies where brokers make transactions as soon as they are available, giving priority to agents who arrived earlier.}  The previous results indicate that a consequence of this strategy is a substantial increase in the fraction of perished agents. We show that if the discount rate is not ``too large'', then for a tuned value of $\alpha$, the $\apatient{\alpha}$ algorithm generates a higher social welfare than the Greedy algorithm. In other words, if agents are not too impatient, a matching algorithm which thickens the market by slightly increasing waiting times delivers higher social welfare than a matching algorithm which simply optimizes the speed of transactions.


Our forth result discusses the incentive-compatible implementation of the $\apatient{\alpha}$ algorithm. 
To implement this algorithm, the planner needs at least partial information about departure times of agents. 
If agents know the urgency of their needs, but the planner does not, they may have incentives to mis-report their departure times so as to hasten their match or to increase their probability of getting matched.  We show that if agents are not too impatient, a dynamic mechanism without transfers can elicit such information.  Specifically, we show that it is arbitrarily close to optimal for agents to report the truth in large markets.


%

\subsection{Related Work}\label{relatedworks}

There have been several studies on dynamic matching in the literature of Computer Science, Economics, and Operations Research that each fit a specific market-place, such as the real estate market, paired kidney exchange, or online advertising. But, to the best of our knowledge, no previous work has offered a general framework for dynamic matching in networked markets, and no previous work has considered stochastic departures.

Kurino \cite{Kur09} and Bloch and Houy \cite{Blo12} study an overlapping generations model of the housing market. In their models, agents have deterministic arrivals and departures. In addition, the housing side of the market is infinitely durable and static, and houses do not have preferences over agents.  In the same context, Leshno \cite{Les12} studies a one-sided dynamic housing allocation problem in which houses arrive stochastically over time. His model is based on two waiting lists and does not include a network structure. In addition, agents remain in the waiting list until they are assigned to a house; i.e., they do not perish.

In the context of live-donor kidney exchanges, \"Unver  \cite{Unv10} studies a dynamic matching model of the market in which agents have multiple types. In his model, agents never perish and he does not analyze the effects of the graph structure on optimal policies explicitly. In the Operations Research and Computer Science literatures, dynamic kidney matching has been extensively studied, see e.g.,  \cite{Zenios2002, Su2005, Awasthi2009, DPS12b}. Perhaps most related to our work is that of Ashlagi, Jaillet, and Manshadi \cite{AJM13} who construct a discrete-time finite-horizon model of dynamic kidney exchange. In their model, one new agent arrives at the pool at each period, but, unlike our model, agents who are in the pool neither perish, nor bear any waiting cost. Their model has two types of agents, one easy to match and one hard to match, which then creates a specific graph structure that fits well to the kidney market.


The problem of online matching has been extensively studied in the literature of online advertising.  In this setting, advertisements are static, but queries arrive adversarially or stochastically over time.  Unlike our model, queries persist in the market for exactly one period.  Karp, Vazirani and Vazirani \cite{karp1990optimal} introduced the problem and designed a randomized matching algorithm. Subsequently, the problem is considered under several arrival models with pre-specified budgets for the advertisers, \cite{mehta2007adwords,goel2008online,feldman2009online,MOS12}.

In contrast to dynamic matching, there are numerous investigations of dynamic auctions and dynamic mechanism design. Parkes and Singh \cite{parkes2003mdp} generalize the VCG mechanism to a dynamic setting. Athey and Segal  \cite{athey2007designing} construct efficient and incentive-compatible dynamic mechanisms for private information settings. Pai and Vohra  \cite{pai2013optimal} and Gallien \cite{gallien2006dynamic} extend Myerson's   optimal auction result \cite{myerson1981optimal} to dynamic environments. We refer interested readers to Parkes  \cite{Parkes2007} for a review of the dynamic mechanism design literature.


\medskip
The rest of the paper is organized as follows. 
\autoref{sec:model} introduces our dynamic matching market model and defines the objective. \autoref{sec:contributions} presents our main contributions; we recommend readers to take a look at this section to see a detailed description of our results without getting into the details of the proofs. Then in \autoref{sec:optimumbounds}, we analyze two optimal policies as benchmarks and provide analytical bounds on their performance. In \autoref{sec:markovchain} we model our algorithms as Markov Chains and bound the mixing time of the chains. 
\autoref{twoalgo} goes through a deep analysis of the Greedy algorithm, the Patient algorithm, and the $\apatient{\alpha}$ algorithm and bounds their performance. In \autoref{sec:welfare}, we take the waiting cost into account and bound the social welfare under different algorithms.  \autoref{mechanism} considers the case where the urgency of an agent's needs is private information, and exhibits a truthful direct revelation mechanism.   \autoref{conclusion} suggests generalizing extensions for the model and concludes.  

\section{The Model}\label{sec:model}
In this section, we provide a stochastic continuous-time model for a bilateral matching market that runs in the interval $[0,T]$. 
Agents arrive at the market at rate $m$ according to a Poisson process, i.e.,
in any interval $[t,t+1]$, $m$ new agents enter the market in expectation. Throughout the paper we assume $m\geq 1$. For $t\geq 0$, let $\A_t$ be the set of the agents in our market at time $t$, and let $Z_t:=|A_t|$.
We refer to $\A_t$ as the {\em pool} of the market.
We start by describing the evolution of $\A_t$ as a function of $t\in [0,T]$. Since we are interested in the limit behavior of $\A_t$, without loss of generality,
we may assume $\A_0=\emptyset$.
We use $\A^n_t$ to denote the set of agents who enter the market at time $t$.
Note that with probability 1, $|\A^n_t|\leq 1$.

An agent $\a\in\A_t$ leaves the market at time $t$, if either of the following two events occur at time~$t$:
\begin{itemize}
\item $\a$ is matched with another agent $\b\in \A_t$,
\item $\a$ becomes {\em critical}.
\end{itemize}

Each agent becomes critical according to an independent Poisson process with rate $\lambda$.  This implies that, if an agent $\a$ enters the market at time $t_0$, then she becomes critical at some time $t_0+X$ where $X$ is an exponential random variable with parameter $\lambda$. Therefore, for any matching algorithm, $\a$ leaves the market at some time $t_1$ where $t_0\leq t_1\leq t_0+X$.
The {\em sojourn} of $\a$ is the length of the interval that $\a$ is in the pool, i.e., $s(a):=t_1-t_0$.
An agent $\a$ {\em perishes} if $\a$ leaves the market unmatched.\footnote{The interpretation of perishing is grimly obvious in the case of kidney exchanges; but we intend this as a term of art.}
We use $\A^c_t\subseteq \cup_{0\leq \tau\leq t} A_\tau$ to denote the set of agents that are critical at time $t$. 
Note that for any $t\geq 0$, with probability 1, $|\A^c_t|\leq 1$.

An agent receives zero utility if she leaves the market unmatched. If she is matched,
she receives a utility of 1 discounted at rate $\delta$.
Formally,
$$ u(a) := \begin{cases}
e^{-\delta s(a)} & \text{if $a$ is matched}\\
0 & \text{otherwise.}
\end{cases}$$

Each pair of distinct agents regards the bilateral transaction between them as \emph{acceptable}
  with probability $d/m$, independent of any other pair of agents in the market.
  For any $0\leq t$, let $E_t \subseteq \A_t\times \A_t$ be the set of acceptable bilateral transactions between the agents in the market at time $t$, and let $G_t=(\A_t, E_t)$.
  Note that if $\a,\b\in \A_t$ and $\a,\b\in \A_{t'}$, then $(\a,\b)\in E_t$ if and only if $(\a,\b)\in E_{t'}$,
  \textit{i.e.} the acceptable bilateral transactions are persistent  throughout the process.
  For an agent $\a\in \A_t$ we use $N_t(\a)\subseteq \A_t$ to denote the set of neighbors of $\a$ in $G_t$.
  It follows that, if the planner does not match any agents, then for any fixed $t\geq 0$,
  $G_t$ is distributed as an Erd\"os-R\'eyni graph with parameter $d/m$ 
and $d$ is the average degree of the agents \cite{ER60}. To make our model interesting, throughout the paper we assume $0\leq d\leq m$.


Let $\A=\cup_{t\leq T} \A^n_{t}$, let $E\subseteq \A\times \A$ be the set of acceptable transactions between agents in $A$\footnote{Note that $E \supseteq \cup_{t\leq T} E_t$, and the two sets are not typically equal, since two agents may find it acceptable to transact, even though they are not in the pool at the same time because one of them was matched earlier.}, and let $G=(\A,E)$. Observe that any realization of the above stochastic process is uniquely defined given  variables $A^n_t,A^c_t$ for all $t\geq 0$ and the set of acceptable transactions, $E$.
We call a vector $(m,d,\lambda)$ a {\em dynamic matching market}. It turns out that without loss of generality (by normalizing $m$ and $d$) we can assume $\lambda=1$ (see \autoref{lem:equivalence} for details). So, throughout the paper, unless otherwise specified, we assume $\lambda=1$.

\paragraph{Online Algorithms.}
An online algorithm at any time $t_0$ only knows $G_t$ for $t\leq t_0$ and does not know anything about $G_{t'}$
for $t'>t_0$. We enrich our model by letting the online algorithm exploit the knowledge of critical agents at time $t$; nonetheless, we will extend several of our theorems to the case where the online algorithm
does not have this knowledge. As will become clear, this assumption has a significant impact on the performance of any online algorithm.  At any time $t\geq 0$ an online algorithm selects a possibly empty {\em matching} in $G_t$ (recall
that a set of edges $M_t\subseteq E_t$ is a matching if no two edges share the same endpoints).

We emphasize that the random sets $A_t,M_t,E_t,N_t$ and the random variable $Z_t$ are functions of the underlying matching algorithm.  We abuse the notation and do not include the name of the algorithm when we analyze these variables.



\paragraph{The Goal.} 
 The goal of the Planner is then to design an {\em online} algorithm that maximizes the social welfare, i.e., the sum of the utility of all agents in the market.  Let $\alg(T)$ be the set of matched agents by time $T$,
$$ \alg(T) := \{a\in A: a \text{ is matched by \alg}\}.$$
We may drop the $T$ in the notation $\alg(T)$ if it is clear from context. 

We define the social welfare of an online algorithm to be the  expectation of the average of the utility of all agents in the interval $[0,T]$:
$$\W{\alg} := \E{\frac{1}{mT} \sum_{a\in \alg(T)} e^{-\delta s(a)}} 
$$

The goal of the Planner is to choose an online algorithm that maximizes the welfare for large values of $T$ (see \autoref{thm:Greedyanal}, \autoref{thm:criticalanal}, and \autoref{thm:welfarecritical} for the dependency of our results to $T$).

For ease of exposition, we initially assume that $\delta = 0$, i.e. the \emph{waiting cost} is zero.  In this case, the goal of the Planner is to match the maximum number of agents, or equivalently to minimize the number of perished agents.  The {\em loss} of an online algorithm $\alg$ is defined as the ratio of the {\em expected}\footnote{We consider the expected value as a modeling choice. One may also be interested in objective functions that depend on the \emph{variance} of the performance, as well as the expected value. As will be seen later in the paper, the performance of our algorithms are highly concentrated around their expected value, which guarantees that the variance is very small in most of the cases.} number of perished agents to the expected size of $A$, 
$$ \loss{\alg} :=\frac{\E{|A\setminus \alg(T)\setminus A_T|}}{\E{|A|}} = \frac{\E{|A\setminus \alg(T)\setminus A_T|}}{mT}. $$


When we assume $\delta = 0$, we will use the $\mathbf{L}$ notation for the planner's loss function.  When we consider $\delta > 0$, we will use the $\mathbf{W}$ notation for social welfare.

Each  of the above optimization problems can be modeled as a Markov Decision Process (MDP)\footnote{We recommend \cite{Ber00} for background on Markov Decision Processes.} that is defined as follows.
The state space is the set of pairs $(H, B)$ where $H$ is any undirected graph of any size,
and if the algorithm knows the set of critical agents, $B$ is a set of at most one vertex of $H$ representing the corresponding critical agent.
The action space
for a given state is the set of matchings on the graph $H$.  Under this conception, an algorithm designer wants to minimize the loss or maximize the social welfare over a time period $T$.

Although this MDP has infinite number of states, with small error one can reduce the state space
to graphs of size at most $O(m)$.
Even in that case, this MDP has an exponential number of states in
$m$, since there are at least $2^{{m\choose 2}}/m!$ distinct graphs of size $m$\footnote{This lower bound is derived as follows:  When there are $m$ agents, there are $m\choose 2$ possible edges, each of which may be present or absent.  Some of these graphs may have the same structure but different agent indices.  A conservative lower bound is to divide by all possible re-labellings of the agents ($m!$).}, so for even moderately large markets\footnote{For instance, for $m=30$, there are more than $10^{98}$ states in the approximated MDP.}, we cannot apply tools from Dynamic Programming literature to find the optimum online matching algorithm.

\paragraph{Optimum Solutions.}
In many parts of this paper we compare the performance of an online algorithm to the performance of
an optimal \emph{omniscient} algorithm. Unlike any online algorithm, the omniscient algorithm
has full information about the future, i.e., it knows the full realization of the graph $G$.\footnote{In computer science, these are equivalently called \emph{offline} algorithms.}
Therefore, it can return the maximum matching in this graph as its output, and thus minimize the fraction of perished agents.
 Let $\optoff(T)$ be the set of matched agents in the {\em maximum} matching of $G$.
The loss function under the omnsicient algorithm at time $T$ is
$$ \loss{\optoff} := \frac{\E{|A\setminus \optoff(T)\setminus A_T|}}{mT} $$
Observe that for any online algorithm, \alg, and any realization of the probability space, we have $ |\alg(T)| \leq |\optoff(T)|$.\footnote{This follows from a straightforward revealed-preference argument: For any realization, the optimum offline policy has the information to replicate any given online policy, so it must do weakly better.}

It is also instructive to study the optimum online algorithm, an online algorithm with {\em unlimited}
computational power. By definition, an optimum online algorithm can solve the exponential-sized state
space Markov Decision Problem and return the corresponding matching. We consider
two different optimum online algorithms, $\opton^c$, the algorithm that knows the set of critical agents at time $t$ (with associated loss $\loss{\opton^c}$),
and $\opton$, the algorithm that does not know these sets (with associated loss $\loss{\opton}$). Let $\alg$ be the loss under any online algorithm that does not know the set of critical agents
at time $t$. It follows that
$$ \loss{\alg} \geq \loss{\opton}\geq \loss{\opton^c} \geq \loss{\optoff}.$$
Note that $|\alg|$ and $|\opton|$ are generally incomparable, and depend on the realization of $G$
we may even have $|\alg| > |\opton|$.  Similarly, let $\alg^c$ be the loss under any online algorithm that knows the set of critical agents
at time $t$. It follows that
$$ \loss{\alg^c} \geq \loss{\opton^c}\geq \loss{\optoff}.$$

\section{Our Contributions}\label{sec:contributions}
In this section, we present our main contributions. We provide high level intuition for the results, without getting into the details. The rest of the paper proves the results that we state here and adds detail. 

We first introduce two simple online algorithms 
and a class of interpolating algorithms. Then, we describe our results by comparing the performance of these algorithms with the optimum solutions that we described above. 

The first  algorithm is the Greedy algorithm, which mimics `match-as-you-go' algorithms used in many real marketplaces.  It delivers maximal matchings at any point in time, without regard for the future.

\begin{definition}[Greedy Algorithm:] If any new agent $\a$ enters the market at time $t$, then
	match her with an arbitrary agent in $N_t(\a)$ whenever $N_t(\a)\neq\emptyset$.
	We use $\loss{\algGreedy}$ and $\W{\algGreedy}$ to denote the loss and the social welfare under this algorithm, respectively.
\end{definition}
Note that	since $|\A^n_t|\leq 1$, in the above definition, we do not need to consider the case where more than one agent enters the market. Observe  that the graph $G_t$ in the Greedy algorithm is (almost) always an empty graph. 
By definition, the Greedy algorithm does not use any information about the set of critical agents.

The second algorithm is a simple online algorithm that preserves two essential characteristics
of $\opton^c$ when $\delta = 0$ (recall that $\opton^c$ is the optimum online algorithm with knowledge of the set of critical agents):
\begin{enumerate}[i)]
\item A pair of agents $\a,\b$ get matched in $\opton^c$
only if one of them is critical.
This property is called the \emph{rule of deferral match}: Since $\delta = 0$, if
$\a,\b$ can be matched and neither of them is critical we can wait and match them later.

\item If an agent $\a$ is critical at time $t$ and $N_t(\a)\neq\emptyset$ then $\opton^c$ matches $\a$.
This property is a corollary of the following simple fact: matching
a critical agent does not increase the number of perished agents in any online algorithm.
\end{enumerate}
Our second algorithm is designed to be the simplest possible online algorithm that satisfies both of the above properties.
\begin{definition}[Patient Algorithm] If an agent $\a$ becomes critical at time $t$, then match her uniformly at random
with an agent in $N_t(\a)$ whenever $N_t(\a)\neq \emptyset$.
We use $\loss{\algcritical}$ and $\W{\algcritical}$ to denote the loss and the social welfare under this algorithm, respectively.\end{definition}

Observe that unlike the Greedy algorithm, here we need access to the set of critical agents at time $t$. We do not intend the timing assumptions about critical agents to be interpreted literally.  An agent's point of perishing represents the point at which it ceases to be socially valuable to match that agent.  Letting the Planner observe the set of critical agents is a modeling convention that represents high-accuracy short-horizon information about agent departures.    An example of such information is the Model for End-Stage Liver Disease (MELD) score, which accurately predicts 3-month mortality among patients with chronic liver disease.  The US Organ Procurement and Transplantation Network gives priority to individuals with a higher MELD  score, following a broad medical consensus that liver donor allocation should be based on urgency of need and not substantially on waiting time. \cite{W03}  Note that the Patient algorithm exploits \emph{only} short-horizon information about urgent cases, as compared to the Omniscient algorithm which has full information of the future.

The third algorithm interpolates between the Greedy and the Patient algorithms. The idea of this algorithm is to
assign independent exponential clocks with rates $1/\alpha$ where $\alpha \in [0, \infty)$ to each agent $\a$. If agent $\a$'s exponential clock ticks, the market-maker attempts to match her. If she has no neighbors, then she remains in the pool until she gets critical, where the market-maker attempts to match her again. 

A technical difficulty with the above matching algorithm is that it is not memoryless; that is because when an agent gets critical and has no neighbors, she remains in the pool. 
Therefore, instead of the above algorithm, we study a slightly different matching algorithm (with a worse loss).  

\begin{definition}[The $\apatient{\alpha}$ algorithm]
Assign independent exponential clocks with rates $1/\alpha$ where $\alpha \in [0, \infty)$ to each agent $\a$. If agent $\a$'s exponential clock ticks or if an agent $\a$ becomes critical at time $t$, match her uniformly at random with an agent in $N_t(\a)$ whenever $N_t(\a)\neq \emptyset$. In both cases, if $N_t(a) = \emptyset$, treat that agent as if she has perished; i.e., never match her again.
We use $\loss{\apatient{\alpha}}$ and $\W{\apatient{\alpha}}$ to denote the loss and the social welfare under this algorithm, respectively.
\end{definition}

It is easy to see that an upper bound on the loss of the $\apatient{\alpha}$ algorithm is an upper bound on the loss of our desired interpolating algorithm.



By definition, an increase in $\alpha$ \emph{increases} the waiting time of the agents as their exponential clocks tick with a slower rate.\footnote{The use of exponential clocks is a modelling convention that enables us to reduce waiting times while retaining analytically tractable Markov properties.} In particular, $\apatient{\infty}$ algorithm is equivalent to the Patient algorithm, and $\apatient{0}$ is equivalent to the Greedy algorithm.


\medskip
In the rest of this section we describe our contributions. To avoid cumbersome notation, we state our results in the large-market long-horizon regime (i.e., as $m\to\infty$ and $T\to\infty$). In the formal versions we will explicitly study the dependency on $(m,T)$. For example,
 we show that the transition of the market to the steady state takes no more than $O(\log(m))$ time units.
In other words, many of the large time effects that we predict in our model can be seen in poly-logarithmic time in the size of the market.

\subsection{Value of Waiting}
What are the advantages of waiting to thicken the market, compared to implementing `static maximal' matches via the Greedy Algorithm? Our first family of results  study  the advantages of this deferral
and see under what circumstances deferral is (substantially) gainful.
First, we show that the Patient Algorithm 
strongly outperforms the Greedy Algorithm. 
More specifically, we show the number of perished agents in $\algcritical$  is exponentially (in $d$) smaller than the number of perished agents in $\algGreedy$.

\begin{theorem}\label{thm.inf1}
For $d\geq 2$, as $T,m\to\infty$,
$$\loss{\algcritical} \lesssim\footnote{We write $A\lesssim B$ to denote that $A\leq B$ up to lower order terms that depend on $m$ and $T$. In particular as $m,T\to\infty$ these terms go to zero.} (d+1)\cdot e^{-d/2}\cdot \loss{\algGreedy}. $$
\end{theorem}
\begin{proofoverview} We show that for large enough values of $T$ and $m$, $(i)$ $\loss{\algcritical} \lesssim e^{-d/2}$ and $(ii)$  $\loss{\opton} \gtrsim 1/(2d+1)$. By the fact that  $\loss{\algGreedy} \geq \loss{\opton}$ (since the Greedy algorithm does not use information about critical agents) the theorem follows immediately. The key idea in proving both parts is to carefully study the distribution of the pool size, $Z_t$, under any of these algorithms.

For $(i)$, we show that pool size under the Patient algorithm is a Markov chain, it has a unique stationary distribution and it mixes rapidly to the stationary distribution (see \autoref{thm:stationarity}). This implies that for $t$ larger than mixing time, $Z_t$ is essentially distributed as the stationary distribution of the Markov Chain. We show that under the stationary distribution, with high probability, $Z_t \in [m/2,m]$. Therefore, any critical agent has no acceptable transactions with probability at most $(1 - d/m)^{m/2} \leq e^{-d/2}$. 
This proves $(i)$ (see \autoref{sec:critical}).

For $(ii)$, note that  any algorithm which lacks the information of critical agents, the expected perishing rate is equal to the pool size because any critical agent  perishes with probability one. Therefore, if the expected pool size is large, the perishing rate is very high and we are done. On the other hand, if the expected pool size is very low, the perishing rate is again very high because there will be many agents with \emph{no} acceptable transactions during their sojourn. We analyze the trade-off between the above two extreme cases and show that even if the pool size is optimally chosen, the loss cannot be less than what we claimed in $(ii)$ (see \autoref{prop:opton}). 
\end{proofoverview}


A market-maker may not be willing to implement the Patient algorithm for various reasons.
First, the cost of waiting is usually not equal to zero; in other words, agents prefer to be matched earlier (we discuss this cost in detail in \autoref{sec:contwelfare}).
Second, the market-maker may only have approximate knowledge of agents' criticality times.  Motivated by this, we study the fraction of perished agents for the interpolating $\apatient{\alpha}$ algorithm. 
The next result shows that when $\alpha$ is not `too small' (i.e., exponential clocks do not tick with a very fast rate), then $\apatient{\alpha}$ algorithm still (strongly) outperforms the Greedy algorithm. 
In other words, a little waiting is substantially better than not waiting at all. 

\begin{theorem}
\label{thm:inf3}
Let $\bar{\alpha} := 1/\alpha + 1$. For $d\geq 2$, as $T, m \to \infty$, 
\begin{eqnarray*}
\loss{\apatient{\alpha}} \lesssim (d+1) \cdot e^{-d/2\bar\alpha}\cdot \loss{\algGreedy}
\end{eqnarray*}
\end{theorem}
\begin{proofoverview} First, we study a slightly modified (and worse) algorithm.
In this algorithm when an agent's exponential clock ticks and she has no neighbors, the algorithm treats her as a perished agent, i.e., she will never be matched. 
Now, by the additivity of the Poisson process, the loss of  this modified algorithm in a $(m,d,1)$ matching market
is equal to the loss of the Patient algorithm in a $(m,d,\bar\alpha)$ matching market, where $\bar\alpha=1/\alpha+1$. 

The next step is the key idea of the proof: we show that a matching market $(m, d, \bar\alpha)$ is \emph{equivalent} to a matching market $(m/\bar\alpha, d/\bar\alpha, 1)$ in the sense that any quantity in these two markets is the same up to a time scale (see \autoref{def:equivalence}). By this fact, the loss of the Patient algorithm
on a $(m,d,\bar\alpha)$ matching market at time $T$ is equal to the loss of Patient algorithm on a $(m/\bar\alpha,d/\bar\alpha,1)$ market at time $\bar\alpha T$. But, we have already upper bounded the latter in our previous results. 
\end{proofoverview}

A numerical example clarifies the significance of this result. 
Consider the case of a kidney exchange market, where $1000$ new patients arrive to the market every year, their average sojourn is $1$ year and they can exchange kidneys with a random patient with probability $\frac{2}{100}$; that is, $d = 20$. The above result for the $\apatient{\alpha}$ algorithm suggests that 
the market-maker can promise to match agents in less than 6 months (in expectation) while the fraction of perished agents is at most 13\% of the Greedy algorithm.

\subsection{Value of Optimization}
Our second family of results study the value of utilizing the underlying network structure.
Let us first give a simple example to show that by exploiting the underlying structure of the network one can improve the expected number of matches. Let $G_t$ be the graph shown in \autoref{fig:arbitraryexample},
and let $\a_2\in\A^c_t.$ Observe that it is strictly better to match $\a_2$ to $\a_1$ as opposed to $\a_3$. So, an algorithm that utilizes the global structure of the underlying network can do strictly better
than an algorithm that decides only based on the local neighborhood of the agents.
Note that Patient and Greedy algorithms make decisions that depend only on the immediate
neighbors of the  agent they are trying to match, so  they cannot differentiate between $a_1$ and $a_3$ in this example.  By contrast, the optimal solutions can use their
unlimited computation power to exploit the structure of the network.

\begin{figure}
\centering
\begin{tikzpicture}
\tikzstyle{every node}=[draw,circle];
\foreach \i/\j in {1/1,2/1,3/2,4/3}{
\path (2*\i,0) node (\i) {$\a_\i$} edge (\j);
}
\end{tikzpicture}
\caption{If $a_2$ gets critical in the above graph it is strictly better to match him to $a_1$ as opposed to $a_3$}
\label{fig:arbitraryexample}
\end{figure}
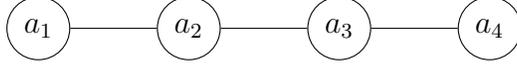

In the following theorem we show that our two algorithms that only depend on the local neighborhood of the agents perform close to the the optimum solutions that make use of the entire  network. 
In particular, among all online algorithms that do not employ 
the criticality information, the Greedy algorithm is one of the best in the sense that its loss is close to $\opton$.
Furthermore, among those algorithms that {\em do} employ the criticality information, the Patient Algorithms is one of the best in the sense that its loss is close to $\optoff$ (and thus $\opton^c$).  Note that $\loss{\algGreedy}$ and $ \loss{\opton}$ are both \emph{fractionally} small in $d$, whereas $\loss{\algcritical}$ and $\loss{\optoff}$ are \emph{exponentially} small in $d$.

\begin{theorem}
\label{thm:network}
For $d\geq 2$, as $T\to\infty$,
\begin{eqnarray*}
\loss{\algGreedy} &\lesssim& (2+1/d)\log{2}\cdot \loss{\opton} \lesssim \frac{\log(2)}{d},\\
\loss{\algcritical} &\lesssim&  \sqrt{\loss{\optoff}\cdot (d+1)/2} \lesssim \frac12\cdot e^{-d/2}.
\end{eqnarray*}
\end{theorem}
\begin{proofoverview} We prove that for large values of $T$ and $m$, $(i)$ $\loss{\algGreedy} \leq \log(2)/d$, and $(ii)$ $\loss{\optoff} \gtrsim e^{-d}/(d+1)$. Note that the theorem follows from these and our results in \autoref{thm.inf1} (that $\loss{\algcritical} \lesssim e^{-d/2}/2$ and  $\loss{\opton} \gtrsim 1/(2d+1)$).

For $(i)$, similar to the Patient algorithm we show that the pool size under the Greedy algorithm is a Markov chain, has a unique stationary distribution and mixes rapidly. Under the stationary distribution,  $Z_t\in [0,\log(2)m/d]$ with high probability. Since $G_t$ under the Greedy algorithm is almost always an \emph{empty} graph,  all critical agents perish, so the perishing rate is at most $\log(2)/d$. 

For $(ii)$, we simply lower bound the fraction of agents who arrive the market at some point in time and have no acceptable transactions to any agent during their sojourn (see \autoref{prop:optoff}). 
\end{proofoverview}

This constitutes an answer to the ``when to match versus whom to match'' question.
Consider the following scenario:
A market maker has limited resources and wants to implement an online matching algorithm. He can employ expert algorithm designers for his particular market and purchase supercomputers to truly exploit the structure of the network. On the other hand, he can fund investigations of the agents' departure probabilities, and then use the $\apatient{\alpha}$ algorithm with a tuned value of $\alpha$. The above result suggests that it is more advantageous that he spends the money on the latter.

A corollary of this result is that waiting and criticality information are complements.  If the Planner does not know about urgent cases, then the Planner cannot substantially reduce losses by waiting: $\loss{\algGreedy}$ is close to $\loss{\opton}$.  If the Planner is not willing to wait, then knowing about urgent cases yields no improvement:  Under the Greedy algorithm, the market is almost always an empty graph, so it cannot address the needs of critical agents.  It is only when waiting and criticality information are combined that losses can be dramatically reduced.

\subsection{Welfare Under Discounting}
\label{sec:contwelfare}

In many practical applications, market brokers employ variants of the Greedy algorithm in order to minimize the waiting time of the agents.
It is now clear that the immediate consequence of this strategy is  a substantial increase in the fraction of perished agents. 
In this part we want to account for the cost of waiting and study online algorithms that optimize social welfare. 

It is clear that if the agents are very impatient (i.e., they have very high waiting costs) the market-maker  is better off implementing the Greedy algorithm. 
On the other hand,  if agents are perfectly patient (i.e., the cost of waiting is zero) the Patient algorithm is substantially better.
Therefore, a natural welfare economics question is as follows:  For which values of $\delta$ is the Patient algorithm (or $\apatient{\alpha}$ algorithm) socially preferred to the Greedy algorithm?

%

Our next result studies social welfare under the Patient, $\apatient{\alpha}$ and Greedy Algorithms.
We show that for small enough $\delta$, there exists a $\apatient{\alpha}$ Algorithm that is socially preferable to the Greedy Algorithm.

\begin{theorem}\label{thm:welfare}
For any $0 \leq \delta \leq \frac{1}{6\log(d)}$ and $d\geq 7$ there exists an $\alpha \geq 0$ such that as $m,T \to \infty,$
$$ \W{\apatient{\alpha}} \geq \W{\algGreedy}. $$
In particular, for $\delta \leq 1/2d$ we have
$$ \W{\algcritical} \geq \W{\algGreedy}.$$
\end{theorem}
\begin{proofoverview} We show that  $(i)$ $\W{\algcritical} \geq \frac{2}{2+\delta}(1 - e^{-d/2})$ and $(ii)$ $\W{\algGreedy} \leq 1-\frac{1}{2d+1}$. As a corollary of our market equivalence result, $(i)$ shows that for $\bar\alpha=1/\alpha+1$, $\W{\apatient{\alpha}} \geq \frac{2}{2+\delta/\bar\alpha}(1 - e^{-d/2\bar\alpha})$.

(ii) can be proved easily just by noting that $1/(2d+1)$ fraction of agents perish in the Greedy algorithm. Therefore, even if all of the matched agents receive a utility of 1 the social welfare is no more than $1-1/(2d+1)$.

The proof of (i) is more involved. The idea is to define a random variable $X_t$ representing the \emph{potential}
utility of the agents in the pool at time $t$, i.e., if all of  agents of $A_t$ get matched immediately, then they receive 
a total utility of $X_t$. It turns out that $X_t$ can be estimated with a small error by studying the evolution of the system
through a differential equation. Given $X_t$ the expected utility of an agent that is matched at time $t$ is exactly $X_t/Z_t$. 
Using our concentration results on $Z_t$, we
can easily compute the expected utility of the agents that are matched in any interval $[t,t+dt]$.
Integrating over all $t\in[0,T]$ proves the claim.
\end{proofoverview}

A numerical example illustrates these magnitudes.  Consider a barter market, where $1000$ new traders arrive at the market every week, their average sojourn is one week, and there is a satisfactory trade between two random agents in the market with probability $\frac{2}{100}$; that is, $d = 10$. Then \autoref{thm:welfare} implies that if the cost associated with waiting for one week is less than $10\%$ of the surplus from a typical trade, then the $\apatient{\alpha}$ algorithm is socially preferred to the Greedy algorithm.

\subsection{Value of Information and Incentive-Compatible Mechanisms}
\label{subsec:contmechanism}

It follows from \autoref{thm:network} that having even short-horizon information about agents' departure times is very valuable, because it substantially reduces perishings.
In particular, since $\loss{\opton} \geq 1/(2d+1)$, under any matching algorithm that does not have any information about agents' departure times, at least $1/(2d+1)$ fraction of the agents will perish.

In many settings, it is plausible that agents have privileged insight into their own departure timings.  In general, agents may have incentives to misreport whether they are critical, in order to increase their chance of getting matched or to decrease their waiting time. We exhibit a truthful mechanism without transfers that elicits such information from agents. 

We assume that agents are fully rational and know the underlying parameters, but they do not observe the actual realization of the stochastic process.  That is, agents observe whether they are critical, but do not observe $G_t$, while the Planner observes $G_t$ but does not observe which agents are critical. Consequently, agents' strategies are independent of the realized sample path. Our results are sensitive to this assumption\footnote{This assumption is plausible in many settings; generally, centralized brokers know more about the current state of the market than individual traders.  Indeed, frequently agents approach centralized brokers because they do not know who is available to trade with them.}; for instance, if the agent knew that she had a neighbor, or knew that the pool at that moment was very large, she would have an incentive under our mechanism to falsely report that she was critical.

The truthful mechanism, $\amechpatient{\alpha}$, is described below. 
\begin{definition}[$\amechpatient{\alpha}$]
Assign independent exponential clocks with rate $1/\alpha$ to each agent $a$, where $\alpha\in[0,\infty)$. If agent's exponential clock ticks or if she  reports becoming critical,
the market-maker attempts to match her to a random neighbor. If the agent has no neighbors, the market-maker treats her as if she has perished, i.e., she will never be matched again.
\end{definition}

Each agent $a$ selects a mixed strategy by choosing a function $c_a(.)$; at the interval $[t,t+dt]$ after her arrival, if she is not yet critical she reports becoming critical with rate $c_a(t) dt$, and when she truly gets critical she reports that immediately.
Our main result in this section asserts that if agents are not too impatient, then the $\amechpatient{\alpha}$ is (incentive-compatible) implementable in the sense that the truthful strategy profile is an $\eps$-mixed strategy Nash equilibrium.\footnote{For a rigorous definition of an $\eps$-mixed strategy Nash equilibrium, see \autoref{def:epsnash}.}


\begin{theorem}
Suppose that the market is at stationary and $d = \polylog(m)$\footnote{$\polylog(m)$ denotes any polynomial function of $\log(m)$.}. Let $\bar\alpha = 1/\alpha + 1$ and $\beta = \bar\alpha(1-d/m)^{m/\bar\alpha}$. Then, for $0 \leq \delta \leq \beta$,  $c_a(t)=0$ for all $a,t$ (i.e., truthful strategy profile) is an $\eps$-mixed strategy Nash equilibrium for $\amechpatient{\alpha}$, where $\eps \to 0$ as $m \to \infty$.

In particular, if $d \geq 2$ and $0 \leq \delta \leq e^{-d/2}$, the truthful strategy profile is an $\eps$-mixed strategy Nash equilibrium for $\amechpatient{\infty}$, where $\eps \to 0$ as $m \to \infty$.
\end{theorem}

\begin{proofoverview}
There is a hidden obstacle in proving that truthful reporting is incentive-compatible:  Even if one assumes that the market is in a stationary distribution at the point an agent enters, the agent's beliefs about pool size may change as time passes.  In particular, an agent makes inferences about the current distribution of pool size, \emph{conditional} on not having been matched yet, and this conditional distribution is different from the stationary distribution. This makes it difficult to compute the payoffs from deviations from truthful reporting.  We tackle this problem by using the concentration bounds from \autoref{prop:criticalconcentration}, and focusing on $\eps$-mixed strategy Nash equilibrium, which allows small deviations from full optimality.

The intuition behind this proof is that an agent can be matched in one of two ways under $\amechpatient{\infty}$: Either she becomes critical, and has a neighbor, or one of her neighbors becomes critical, and is matched to her.  By symmetry, the chance of either happening is the same, because with probability 1, every matched pair consists of one critical agent and one non-critical agent.  When an agent declares that she is critical, she is taking her chance that she has a neighbor in the pool right now.  By contrast, if she waits, there is some probability that another agent will become critical and be matched to her.   
Consequently, for small $\delta$, agents will opt to wait.
\end{proofoverview}

\subsection{Technical Contributions}
As alluded to above, most of our results follow from concentration results on the distribution
of the pool size for each of the online algorithms, \autoref{prop:Greedyconcentration}, \autoref{prop:criticalconcentration}. In this last part we describe  ideas behind these crucial results. 

For analyzing many classes of stochastic processes one needs to prove concentration bounds on  functions defined on the underlying process by means of Central limit
theorem, Chernoff bounds or Azuma-Hoeffding bounds. In our case many of these tools fail. This is because we are interested in proving that for any large  time $t$, a given function is concentrated in an interval whose size depend only on $d,m$ and not $t$. Since $t$ can be significantly larger
than $d,m$ a direct proof fails. 

In contrast we observe that $Z_t$ is  a Markov Chain for a large class of online algorithms. Building
on this observation, first we show that the underlying Markov Chain has a unique stationary distribution and it {\em mixes} rapidly. Then we use the stationary distribution of the Markov Chain to prove our concentration bounds.

However, that is not the end of the story. We do not a closed form expression for the stationary distribution of the chain, because we are dealing with an infinite state space continuous time Markov Chain where the transition rates are complex functions of the states.
Instead, we use the following trick. Suppose we want to prove that $Z_t$ is contained in an interval $[k^*-f(m,d), k^*+f(m,d)]$ for some $k^*\in\N$
with high probability,
where $f(m,d)$ is a function of $m,d$ that does not depend on $t$. 
We consider a sequence of pairs of states $P_1:=(k^*-1,k^*+1), P_2:=(k^*-2,k^*+2),$ etc. We show that  if the Markov Chain
is at any of the states of $P_i$, it is more likely (by an additive function of $m,d$) that it jumps to a state of $P_{i-1}$ as opposed to $P_{i+1}$. 
Using {\em balance equations} and simple algebraic manipulations, this implies that
 the  probability of states in $P_i$ {\em geometrically} decrease as $i$ increases. 
In other words $Z^*$ is concentrated in a small interval around $k^*$.
We believe that this technique can be used in studying other complex stochastic processes.

\section{Performance of the Optimum Solutions}
\label{sec:optimumbounds}
In this section we lower-bound the loss of the optimum solutions in terms of $d$. In particular, we prove the following theorems.

\begin{theorem}
\label{prop:opton}
If $ m > 10d$, then for any $T>0$
$$  \loss{\opton} \geq \frac{1}{2d+1 + d^2/m}.$$
\end{theorem}

\begin{theorem}
\label{prop:optoff}
If $m>10d$, then for any $T>0$,
$$ \loss{\optoff} \geq \frac{e^{-d-d^2/m}}{d+1+d^2/m}$$
\end{theorem}

Before proving the above theorems, it is useful to study the evolution of the system in the case where no agents are ever matched, i.e., the online algorithm does nothing. We later use this analysis in this section, as well as \autoref{sec:markovchain} and \autoref{twoalgo}.

We adopt the notation $\tilde{A}_t$ and $\tilde{Z}_t$ to denote the agents in the pool and the pool size in this case. Observe that by definition for any matching algorithm and any realization of the process,
\begin{equation}
\label{eq:ZtZ}
Z_t \leq \tilde{Z_t}.
\end{equation}
Using the above equation, in the following fact we show that for any matching algorithm $\E{Z_t}\leq m$.
\begin{proposition}
\label{fact:expectedpoolsize}
For any $t_0\geq 0$,
$$ \E{\tilde{Z}_{t_0}}  = (1-e^{-t_0})m.$$
\end{proposition}
\begin{proof}
Let $\a_1,\ldots,\a_K$ be the set of agents who enter the market in the interval $[0,t_0]$ where $K$ is a random variable and the index is arbitrary and not by order of entry.  Let $X_{a_i}$ be the random variable indicating that $\a_i\in \tilde{A}_{t_0}$. Then,
$$ \P{X_{\a_i}} = \int_{t=0}^{t_0} \frac{1}{t_0} \P{s(a)\geq t_0-t} dt = \frac{1}{t_0}\int_{t=0}^{t_0} e^{t-t_0} dt = \frac{1-e^{-t_0}}{t_0}. $$
Therefore, $\E{\tilde{Z}_{t_0}} = \E{\E{\sum_{i=1}^K X_{a_i}|K}} = \E{\frac{K}{t_0}(1-e^{-t_0})} = (1-e^{-t_0})m.$
\end{proof}

\subsection{Loss of the Optimum Online Algorithm}

In this section, we prove \autoref{prop:opton}.
Let $\poolsize$ be the expected pool size of the $\opton$,
$$\poolsize:=\EE{t\sim \unif[0,T]}{Z_t}$$
Since $\opton$ does not know the $\A^c_t$ we can assume that each critical agent perishes with probability 1. Therefore,
\begin{equation}
\label{eq:averageloweropton}
\loss{\opton} = \frac{1}{m\cdot T} \E{\int_{t=0}^T Z_t dt} = \frac{\poolsize T}{m T} = \poolsize/m.
\end{equation}
To finish the proof we need to lower bound $\poolsize$ by $m/(2d+1)$. We
provide an indirect proof by showing a lower-bound on $\loss{\opton}$ which in turn
lower-bounds $\poolsize$.

Our idea is to lower-bound the probability that an agent does not have any acceptable transactions throughout her sojourn, and this directly gives a lower-bound on $\loss{\opton}$.
Fix an agent $\a\in A$. Conditioned on  $\a\in A$, $\a$ enters the market at a time $t_0 \sim \unif[0,T]$, and
$s(\a)=t$.
Therefore,
\begin{eqnarray*}
 \P{N(\a) = \emptyset} &=& \int_{t=0}^{\infty} \P{s(a) = t} \cdot \E{(1-d/m)^{|\A_{t_0}|}} \cdot \E{(1-d/m)^{ |\A^n_{t_0,t+t_0}|}}dt\\
 &\geq & \int_{t=0}^\infty e^{-t} \cdot (1-d/m)^{\E{Z_{t_0}}} \cdot (1-d/m)^{\E{|A^n_{t_0,t+t_0}|}}dt\\
 &=&  \int_{t=0}^\infty e^{-t} \cdot (1-d/m)^{\poolsize} \cdot (1-d/m)^{mt} dt
 \end{eqnarray*}
where the second inequality follows by the Jensen's inequality.
Since $d/m<1/10$, $1-d/m \geq e^{-d/m -d^2/m^2}$,
\begin{equation}
\label{eq:optonlowerprob}
\loss{\opton} \geq \P{N(a) =\emptyset} \geq e^{-\poolsize(d/m+d^2/m^2)} \int_{t=0}^\infty e^{-t(1+d+d^2/m)}dt \geq \frac{1-\poolsize (1+d/m)d/m}{1+d+d^2/m}
\end{equation}

Putting \eqref{eq:averageloweropton} and \eqref{eq:optonlowerprob} together, for $\beta:=\poolsize d/m$ we get
$$ \loss{\opton} \geq \max\{ \frac{1-\beta(1+d/m)}{1+d+d^2/m}, \frac{\beta}{d}\} \geq \frac{1}{2d+1+d^2/m}$$
where  the last inequality follows by letting $\beta=\frac{d}{2d+1+d^2/m}$ be the minimizer of the middle expression.

\subsection{Loss of the Omniscient Algorithm}

In this section, we prove \autoref{prop:optoff}. This demonstrates that, in the high-information setting, \emph{no} policy can yield a faster-than-exponential decrease in losses, as a function of the average degree of each agent.

The proof is very similar to \autoref{prop:opton}. Let $\poolsize$ be the expected pool size of the $\optoff$,
$$ \poolsize:=\EE{t\sim\unif[0,T]}{Z_t}.$$
By \eqref{eq:ZtZ} and \autoref{fact:expectedpoolsize},
$$\poolsize\leq \EE{t\sim \unif[0,T]}{\tilde{Z}_t} \leq  m.$$
Note that \eqref{eq:averageloweropton} does
not hold in this case because the optimum offline algorithm knows the set of critical agents at time $t$.

Now, fix an agent $\a\in \A$, and let us lower-bound the probability that $N(\a)=\emptyset$.
Say $\a$ enters the market at time $t_0\sim\unif[0,T]$ and $s(\a)=t$, then
\begin{eqnarray*}
\P{N(\a)=\emptyset} &=& \int_{t=0}^\infty \P{s(a)=t}\cdot \E{(1-d/m)^{Z_{t_0}}} \cdot \E{(1-d/m)^{|A^n_{t_0,t+t_0}|}} dt \\
&\geq & \int_{t=0}^\infty e^{-t} (1-d/m)^{\poolsize+mt} dt \geq \frac{e^{-\poolsize(1+d/m)d/m}}{1+d+d^2/m}
\geq \frac{e^{-d-d^2/m}}{1+d+d^2/m}.
\end{eqnarray*}
where the first inequality uses the Jensen's inequality and the second inequality uses
the fact that $d/m < 1/10$.

\section{Modeling an Online Algorithm as a Markov Chain}\label{sec:markovchain}
\subsection{Background}
In this section we show that in both of the Patient and Greedy algorithms the random processes $Z_t$ are Markovian, have unique stationary distributions and mix rapidly.
Before getting into the details we provide a brief overview on continuous time Markov Chains. We refer interested readers to \cite{Nor98,LPW06} for detailed discussions.

Let $Z_t$ be a continuous time Markov Chain on the non-negative integers ($\mathbb{N}$) that starts from state zero.
For any two states $i,j\in\N$, we assume that the rate of going from $i$ to $j$ is $r_{i\to j} \geq 0$.
The rate matrix $Q \in \mathbb{N}\times\mathbb{N}$ is defined as follows,
\begin{equation*}
Q(i,j):=\begin{cases}
r_{i\to j} & \text{if $i\neq j$,}\\
\sum_{k\neq i} -r_{i\to k} & \text{otherwise.}
\end{cases}
\end{equation*}
Note that, by definition, the sum of the entries in each row of $Q$ is zero.
It turns out that (see e.g., \cite[Theorem 2.1.1]{Nor98}) the transition probability in $t$ units of time is,
$$ e^{tQ} = \sum_{i=0}^\infty \frac{t^i Q^i}{i!}.$$
Let $P_t:=e^{tQ}$ be the transition probability matrix of the Markov Chain in $t$ time units. 
It follows that,
\begin{equation}
\label{eq:dtzt}
 \frac{d}{dt} P_t = P_t Q.
 \end{equation}
 In particular, in any infinitesimal time step $dt$, the chain moves based on $Q\cdot dt$. 
 
A Markov Chain is irreducible if for any pair of states $i,j\in\N$, $j$ is reachable from $i$ with a non-zero probability. 
Fix a state $i\geq 0$, and suppose that $Z_{t_0}=i$, and let $T_1$
be the first jump out of $i$ (note that  $T_1$ is distributed as an exponential random variable).
State $i$ is {\em positive recurrent} iff
\begin{equation}
	\E{\inf \{ t \geq T_1: Z_t = i\} | Z_{t_0}=i} < \infty
\end{equation}
It follows by the ergodic theorem that a (continuous time) Markov Chain has a unique stationary distribution if and only if it has a positive recurrent state \cite[Theorem 3.8.1]{Nor98}. 
Let  $\pi:\N\to\R_+$ be  the stationary distribution of a Markov chain.
It follows by the definition that for any $t\geq 0$, $P_t = \pi P_t$. 
The {\em balance equations} of a Markov chain  say that for any $S\subseteq \N$,
\begin{equation}
\label{eq:balanceequations}
 \sum_{i\in S,j\notin S} \pi(i) r_{i\to j} = \sum_{i\in S,j\notin S} \pi(j) r_{j\to i}.
 \end{equation}

Let $z_t(.)$ be the distribution of $Z_t$ at time $t\geq 0$, i.e., $z_t(i) := \P{Z_t=i}$ for any integer $i\geq 0$.
For any $\eps>0$, we define the mixing time (in total variation distance) of this Markov Chain as follows,
\begin{equation}
\label{eq:mixingtimedef}
 \tau_{\mix}(\eps) = \inf\Big\{t: \norm{z_t - \pi}_{\TV} := \sum_{k=0}^\infty |\pi(k) - z_t(k)| \leq \eps\Big\}.
 \end{equation}

\subsection{Markov Chain Characterization}
The following is the main theorem of this section. We show that this Markov Chain mixes in time $O(\log(m))$.
\begin{theorem}
\label{thm:stationarity}
For the Patient and Greedy algorithms and any $0 \leq t_0 < t_1$,
$$ \P{Z_{t_1} | Z_t \text{ for } {0\leq t<t_1}} = \P{Z_{t_1} | Z_t \text{ for } t_0 \leq t < t_1}.  $$
The corresponding Markov Chains have unique stationary distributions
 and mix in
time $O(\log(m)\log(1/\eps))$ in total variation distance,
$$\tau_{\mix}(\eps) \leq O(\log(m)\log(1/\eps)).$$
\end{theorem}

First, we argue that $Z_t$ is Markovian for the Patient and Greedy algorithms.  This follows from the following simple observation.

\begin{proposition}
\label{Z_suff_stat}
Under either of Greedy or Patient algorithms, for any $t\geq 0$, conditioned on $Z_t$, the distribution of $G_t$ is uniquely defined. So, given $Z_t$, $G_t$ is conditionally independent of $Z_{t'}$ for $t'<t$.
\end{proposition}

\begin{proof}
Under the Greedy algorithm, at  any time $t\geq 0$,
$|E_t|=0$.
Therefore, conditioned on $Z_t$, $G_t$ is an empty graph with $|Z_t|$ vertices.

On the other hand, the Patient algorithm's never looks at the edges between non-critical agents, so the algorithm is oblivious to these edges.
It follows that under the Patient algorithm, for any $t \geq 0$, conditioned on $Z_{t}$,   $G_t$ is an Erd\"os-R\'eyni random graph with $|Z_t|$ vertices and parameter $d/m$.
\end{proof}

It follows from the above proposition that 
$Z_t$ is Markovian under the Greedy Algorithm and the  Patient Algorithm.

\subsection{Stationary Distributions: Existence and Uniqueness}

In this part we show that the Markov Chain on $Z_t$ has a unique stationary distribution under each of the Greedy and Patient algorithms.
In the last section we proved that $Z_t$ is indeed a Markov chain on non-negative integers ($\N$)
that starts from state zero.  

First, we show that the Markov Chain is irreducible.
First note that every state $i>0$ is reachable from state $0$ with a non-zero probability. It is sufficient that $i$ agents arrive at the market with no acceptable bilateral transactions.
On the other hand, state $0$ is reachable from any $i>0$ with a non-zero probability. It is sufficient that all of the $i$ agents in the pool become critical and no new agents arrive at the market.
So $Z_t$ is an {\em irreducible} Markov Chain.

 Therefore, by the ergodic theorem  it has a unique stationary distribution if and only if it has a positive recurrent state \cite[Theorem 3.8.1]{Nor98}. 
In the rest of the proof we show that the state zero is positive recurrent.
By \eqref{eq:ZtZ} $Z_t=0$ if $\tilde{Z}_t=0$. So, it is sufficient to show
\begin{equation}
\label{eq:positiverecurrentempty}
 \E{\inf\{ t\geq T_1: \tilde{Z}_t= 0\} | \tilde{Z}_{t_0}=0} < \infty.
 \end{equation}
It follows that $\tilde{Z}_t$ is just a continuous time birth-death process on $\N$ with
the following transition rates,
\begin{equation}
\label{eq:emptyrates} \tilde{r}_{k\to k+1} = m \text{ and } \tilde{r}_{k\to k-1} := k
\end{equation}
It is well known (see e.g. 
\cite[p.~249-250]{GS92}) that $\tilde{Z}_t$ has a stationary distribution if and only if
$$ \sum_{k=1}^\infty \frac{\tilde{r}_{0\to 1} \tilde{r}_{1\to 2} \ldots \tilde{r}_{k-1\to k}}{\tilde{r}_{1\to 0}\ldots \tilde{r}_{k\to k-1}} < \infty.$$
Using \eqref{eq:emptyrates} we have
\begin{equation*}
\sum_{k=1}^\infty \frac{\tilde{r}_{0\to 1} \tilde{r}_{1\to 2} \ldots \tilde{r}_{k-1\to k}}{\tilde{r}_{1\to 0}\ldots \tilde{r}_{k\to k-1}} = \sum_{k=1}^\infty \frac{m^k}{k!} = e^m - 1 < \infty
\end{equation*}
Therefore, the birth death process has a unique stationary distribution and the state $0$ is positive recurrent. This proves \eqref{eq:positiverecurrentempty} so $Z_t$ is an ergodic Markov Chain.

\subsection{Upper bounding the Mixing Time}
In this part we complete the proof of \autoref{thm:stationarity} and we upper-bound
the mixing of Markov Chain $Z_t$ for the Greedy and Patient algorithms. 
Let $\pi(.)$ be the stationary distribution of the Markov Chain.
\subsubsection{Mixing time of the Greedy Algorithm}
We start with the Greedy algorithm. We use the {\em coupling} technique (see \cite[Chapter 5]{LPW06}). Suppose
we have two Markov Chains $Y_t, Z_t$ (with different starting distributions) 
each running the Greedy algorithm. 
We define  a joint Markov Chain $(Y_t,Z_t)_{t=0}^\infty$ with the property that projecting on either of $Y_t$ and $Z_t$ we see the stochastic process of Greedy algorithm, and that
they stay together at all times after their first simultaneous visit to a single state, i.e.,
$$ \text{if } Y_{t_0}=Z_{t_0}, \text{ then } Y_t=Z_t \text{ for } t\geq t_0.$$

Next we define the joint chain. We define this  chain such that for any $t\geq t_0$, $|Y_t-Z_t| \leq |Y_{t_0}-Z_{t_0}|$.
Assume that $Y_{t_0}=y, Z_{t_0}=z$ at some time $t_0\geq 0$, for $y,z\in\N$. Without loss of generality assume $y < z$ (note that if $y=z$ there is nothing to define). Consider any arbitrary labeling of the agents in the first pool with $a_1,\ldots,a_y$, and in the second pool with $b_1,\ldots,b_z$. Define $z+1$ independent exponential clocks such that the first $z$ clocks have rate 1, and the last one has rate $m$. 
If the $i$-th clock ticks for $1\leq i\leq y$, then both of $a_i$ and $b_i$ become critical (recall that in the Greedy algorithm the critical agent leaves the market right away). If $y< i\leq z$, then  $b_i$ becomes critical, and if $i=z+1$ new agents $a_{y+1}, b_{z+1}$ arrive to the markets. In the latter case we need to draw  edges between the new agents and those currently in the pool. 
We use $z$ independent coins each with parameter $d/m$. We use the first $y$ coins to decide simultaneously on the potential transactions $(a_i,a_{y+1})$ and $(b_i,b_{z+1})$ for $1\leq i\leq y$, and the last $z-y$ coins for the rest.
This implies that for any $1\leq i\leq y$, $(a_i,a_{y+1})$ is an acceptable transaction iff $(b_i,b_{z+1})$ is acceptable. 
Observe that if $a_{y+1}$ has at least one acceptable transaction then so has $b_{z+1}$
but the converse does not necessarily hold.

It follows from the above construction that $|Y_t-Z_t|$ is a non-increasing function of $t$. 
Furthermore, this value decreases when either of the agents $b_{y+1},\ldots,b_{z}$ become critical (we note that this value may also decrease when a new agent arrives but we do not exploit this situation here).
Now suppose $|Y_0-Z_0|=k$. It follows that the two chains arrive to the same state when all of the $k$ agents that are not in common become critical. This has the same distribution as the maximum of $k$ independent exponential random variables with rate 1. Let $E_k$ be a random variable that is the maximum of $k$ independent exponentials of rate 1. For any $t\geq 0$,
$$ \P{Z_t \neq Y_t} \leq \P{E_{|Y_0-Z_0|}\leq  t} = (1-e^{-t})^{|Y_0-Z_0|}.$$

Now, we are ready to bound the mixing time of the Greedy algorithm. Let $z_t(.)$
be the distribution of the pool size at time $t$ when there is no agent in the pool at time $0$
and let $\pi(.)$ be the stationary distribution. 
Fix $0<\eps<1/4$, and let $\beta\geq 0$ be a parameter that we fix later.
Let $(Y_t,Z_t)$ be the joint Markov chain that we constructed above where $Y_t$ is started at the stationary distribution and $Z_t$ is started at state zero. 
Then,
\begin{eqnarray*}
\norm{z_t - \pi}_{\TV} \leq \P{Y_t\neq Z_t}& =& \sum_{i=0}^\infty \pi(i) \P{Y_t\neq Z_t | Y_0=i}\\
&\leq& \sum_{i=0}^\infty \pi(i) \P{E_i > t}\\
&\leq& \sum_{i=0}^{\beta m/ d} (1-(1-e^{-t})^{\beta m/d}) + \sum_{i=\beta m/ d}^\infty \pi(i) 
\leq  \frac{\beta^2 m^2}{d^2} e^{-t} + 2 e^{-m(\beta-1)^2/2d}
\end{eqnarray*}
where the last inequality follows by equation \eqref{eq:bernoulli} and \autoref{prop:Greedyconcentration}.
Letting $\beta = 1+\sqrt{2\log(2/\eps)}$ and $t=2\log(\beta m/d)\cdot \log(2/\eps)$ 
we get $\norm{z_t-\pi}_{\TV}\leq \eps$ that proves the theorem.

\subsubsection{Mixing time of the Patient Algorithm}
It remains to bound the mixing time of the Patient algorithm. 
The construction of the joint Markov Chain is very similar to the above construction except some caveats. Again, suppose $Y_{t_0}=y$ and $Z_{t_0}=z$ for $y,z\in\N$ and $t_0\geq 0$
and that $y<z$. Let $a_1,\ldots,a_y$ and $b_1,\ldots,b_z$ be a labeling of the agents. 
We consider two cases.
\begin{enumerate}[{Case} 1)]
\item $z>y+1$. In this case the construction is essentially the same as the Greedy algorithm.
The only difference is that we toss random coins to decide on acceptable bilateral transactions at the time that an agent becomes critical (and not at the time of arrival). It follows that when new agents arrive the size of each of the pools increase by 1 (so the difference remains unchanged). If any of the agents $b_{y+1},\ldots,b_z$ become critical then the size of second pool decrease  by 1 or 2 and so is the difference of the pool sizes. 
\item $z=y+1$. In this case we define a slightly different coupling. This is because, for some parameters and starting values, the Markov chains may not visit the same state for a long time for the coupling defined in Case 1 . If $z \gg m/d$, then with a high probability any critical agent gets matched. Therefore, the magnitude of $|Z_t-Y_t|$ does not quickly decrease (for a concrete example, consider the case where  $d=m$, $y=m/2$ and $z=m/2+1$).
Therefore, in this case we change the coupling.
We use $z+2$ independent clocks where the first $z$ are the same as before, i.e., they have rate 1 and when the $i$-th clock ticks $b_i$ (and $a_i$ if $i\leq y$) become critical.
The last two clocks have rate $m$, when the $z+1$-st clock ticks a new agent arrives to the first pool and when $z+2$-nd one ticks a new agent arrives to the second pool.
\end{enumerate}
Let $|Y_0-Z_0|=k$. 
By  the above construction  $|Y_t-Z_t|$ is a decreasing function of $t$ unless $|Y_t-Z_t|=1$. In the latter case this difference goes to zero if a new agent arrives to the smaller pool and it increases if a new agent arrives to the bigger pool. 
Let $\tau$ be the first time $t$ where $|Y_t-Z_t|=1$. Similar to the Greedy algorithm, the event $|Y_t-Z_t|=1$ occurs if the second to maximum of $k$ independent exponential random variables with rate 1 is at most t. Therefore,
$$ \P{\tau \leq t} \leq \P{E_k\leq t} \leq (1-e^{-t})^k$$

Now, suppose $t\geq \tau$; we  need to bound the time it takes to make the difference zero. First, note that after time $\tau$ the difference is never more than 2. Let $X_t$ be the (continuous time) Markov Chain illustrated in \autoref{fig:3statechain} and suppose $X_0=1$.
Using $m\geq 1$, it is easy to see that if $X_t=0$ for some $t\geq 0$, then $|Y_{t+\tau}-Z_{t+\tau}|=0$ (but the converse is not necessarily true). It is a simple exercise that for $t\geq 8$,
\begin{equation}
\label{eq:3statechain}
\P{X_t \neq 0} = \sum_{k=0}^\infty \frac{e^{-t} t^k}{k!} 2^{-k/2}\leq \sum_{k=0}^{t/4} \frac{e^{-t} t^k}{k!} + 2^{-t/8} \leq 2^{-t/4} + 2^{-t/8}. 
\end{equation}

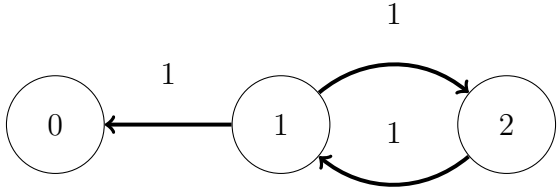
\begin{figure}[t]
\centering
\begin{tikzpicture}[inner sep=0, minimum size=13mm]
\tikzstyle{every node}=[circle,draw];
\node at (0,0) (a) {$0$};
\node at (3,0) (b) {$1$};
\node at (6,0) (c) {$2$};
\draw[->,line width=1.5pt]  (b)  edge node [above,draw=none] {$1$} (a)
(c) [bend left=40] edge node [above=0.5,draw=none] {$1$} (b)
(b) [bend left=40] edge node [above,draw=none] {$1$} (c);
\end{tikzpicture}
\caption{A three state Markov Chain used for analyzing the mixing time of the Patient algorithm.}
\label{fig:3statechain}
\end{figure}
Now, we are ready to upper-bound the mixing time of the Patient algorithm.
Let $z_t(.)$ be the distribution of the pool size at time $t$ where there is no agent at time 0,
and let $\pi(.)$ be the stationary distribution. Fix $\eps >0$, and let $\beta \geq 2$ be a parameter that we fix later. 
Let $(Y_t,Z_t)$ be the joint chain that we constructed above where $Y_t$ is started at the 
stationary distribution and $Z_t$ is started at state zero.
\begin{eqnarray*}
\norm{z_t-\pi}_{\TV} \leq \P{Z_t\neq Y_t} &\leq& \P{\tau \leq t/2} + \P{X_t\leq t/2}\\
&\leq & \sum_{i=0}^\infty \pi(i) \P{\tau\leq t/2 | Y_0=i} + 2^{-t/8+1} \\
&\leq & 2^{-t/8+1}+ \sum_{i=0}^\infty \pi(i) (1-(1-e^{-t/2})^i) \\
&\leq & 2^{-t/8+1} + \sum_{i=0}^{\beta m} (i t/2) + \sum_{i=\beta m}^\infty \pi(i) 
\leq 2^{-t/8+1} + \frac{\beta^2 m^2 t}{2} + 6 e^{-(\beta-1)m/3}.
\end{eqnarray*}
where in the second to last equation we used equation \eqref{eq:bernoulli} and in the last equation we used \autoref{prop:criticalconcentration}.
Letting $\beta=10$ and $t=8\log(m)\log(4/\eps)$ implies that $\norm{z_t-\pi}_{\TV}\leq \eps$ which proves \autoref{thm:stationarity}.

\section{Performance Analysis}\label{twoalgo}
In this section we upper bound $\loss{\algGreedy}$ and $\loss{\algcritical}$ as a function of $d$, 
and we upper bound $\loss{\apatient{\alpha}}$ as a function of $d$ and $\alpha$.
We prove the following three theorems.
\begin{theorem}
\label{thm:Greedyanal}
For any $\eps \geq 0$ and $T>0$,
\begin{eqnarray}
\label{eq:Greedyupper}
\loss{\algGreedy} \leq   \frac{\log(2)}{d} + \frac{\tau_{\mix}(\eps)}{T} + 6\eps  + O\Big(\frac{\log(m/d)}{\sqrt{dm}}\Big),
\end{eqnarray}
where $\tau_{\mix}(\eps) \leq 2\log(m/d)\log(2/\eps)$.
\end{theorem}

\begin{theorem}
\label{thm:criticalanal}
For any $\eps > 0$ and $T>0$,
\begin{eqnarray}
\label{eq:criticalupper}
\loss{\algcritical} \leq \max_{z\in [1/2,1]}\Big(z + \tilde{O}(1/\sqrt{m})\Big)e^{-zd} + \frac{\tau_{\mix}(\eps)}{T} + \frac{\eps m}{d^2} + 2/m,
\end{eqnarray}
where $\tau_{\mix}(\eps) \leq 8\log(m)\log(4/\eps)$.
\end{theorem}


\begin{theorem}
\label{thm:alphaanal}
Let $\bar{\alpha} := 1/\alpha + 1$. For any $\eps > 0$ and $T>0$,
\begin{eqnarray*}
\label{eq:alphaupper}
\loss{\apatient{\alpha}} \leq \max_{z\in [1/2,1]}\Big(z + \tilde{O}(\sqrt{\bar{\alpha} /m})\Big)e^{-zd/\bar{\alpha}} + \frac{\tau_{\mix}(\eps)}{\bar\alpha T} + \frac{\eps m\bar{\alpha} }{d^2} + 2\bar{\alpha} /m,
\end{eqnarray*}
where $\tau_{\mix}(\eps) \leq 8\log(m/\bar\alpha)\log(4/\eps).$
\end{theorem}

We will prove \autoref{thm:Greedyanal} in \autoref{sec:Greedy}, \autoref{thm:criticalanal} in \autoref{sec:critical}
and 
\autoref{thm:alphaanal} in \autoref{sec:alphaanal}.

\subsection{Loss of the Greedy Algorithm}
\label{sec:Greedy}
In this part we upper bound $\loss{\algGreedy}$.
We crucially exploit the fact that  $Z_t$ is a Markov
Chain and has a unique stationary distribution, $\pi: \mathbb{N} \to \mathbb{R}_+$ (see \autoref{thm:stationarity} for proof).

Let $\poolsize:=\EE{Z\sim\mu}{Z}$ be the expected size of the pool under the stationary distribution
of the Markov Chain on $Z_t$.
First, observe that if the Markov Chain on $Z_t$ is mixed, then the agents perish at the rate of
$\poolsize$. Roughly speaking, if we run the Greedy algorithm for a sufficiently long time then Markov Chain on size of the pool mixes and we get $\loss{\algGreedy} \approx \frac{\poolsize}{m}$.
This observation is made rigorous in the following lemma. Note that as $T$ and $m$ grow, the first three terms become negligible.

\begin{lemma}
\label{lem:algGreedyEZ}
For any $\eps>0$, and $T>0$,
$$ \loss{\algGreedy} \leq \frac{\tau_{\mix}(\eps)}{T} + 6\eps + \frac{1}{m}2^{-6m} + \frac{\EE{Z\sim \pi}{Z}}{m}.$$
\end{lemma}
\begin{proof}
By \autoref{fact:expectedpoolsize}, $\E{Z_t} \leq m$ for all $t$, so
\begin{eqnarray} \loss{\algGreedy} = \frac{1}{m\cdot T} \E{\int_{t=0}^T Z_t dt}
&=& \frac{1}{mT} \int_{t=0}^T \E{Z_t} dt \nonumber \\
 &\leq&
\frac{1}{mT} m \cdot \tau_{\mix}(\eps) + \frac{1}{mT}\int_{t=\tau_{\mix}(\eps)}^T \E{Z_t} dt
\label{eq:tmixalgGreedy}
\end{eqnarray}
where  the second equality uses the linearity of expectation. 
Let $\tilde{Z}_t$ be the number of agents in the pool at time $t$ when we do not match any pair of agents. 
By \eqref{eq:ZtZ},
$$\P{Z_t\geq i} \leq \P{\tilde{Z}_t \geq i}.$$
Therefore, for $t\geq \tau_{\mix}(\eps)$,
\begin{eqnarray}
\E{Z_t} = \sum_{i=1}^\infty \P{Z_t\geq i} &\leq&
 \sum_{i=0}^{6m} \P{Z_t\geq i} + \sum_{i=6m+1}^\infty \P{\tilde{Z}_t\geq i}\nonumber \\
 &\leq &
 \sum_{i=0}^{6m} (\PP{Z\sim \pi}{Z\geq i} +\eps) + \sum_{i=6m+1}^\infty \sum_{\ell=i}^\infty \frac{m^\ell}{\ell!} \nonumber \\
 &\leq &
 \EE{Z\sim \pi}{Z} + \eps 6m + \sum_{i=6m+1}^\infty 2\frac{m^i}{i!} \nonumber \\
 &\leq& \EE{Z\sim\pi}{Z} + \eps 6m + \frac{4m^{6m}}{(6m)!} \leq \EE{Z\sim \pi}{Z}+\eps 6m + 2^{-6m}.\label{eq:ZptalgGreedy}
\end{eqnarray}
where the second inequality uses $\P{\tilde{Z}_t = \ell} \leq m^\ell/\ell!$ that is proved below
and the last inequality follows by the Stirling's approximation of $(6m)!$\footnote{Stirling's approximation states that
$$ n! \geq \sqrt{2\pi n} \Big(\frac{n}{e}\Big)^n. $$
}.

Putting \eqref{eq:tmixalgGreedy} and \eqref{eq:ZptalgGreedy} proves the lemma.
\begin{claim}
For any $t_0> 0$,
$$ \P{\tilde{Z}_{t_0}=\ell}  \leq \frac{m^\ell}{\ell!}.$$
\end{claim}
\begin{proof}
Let $K$ be a random variable indicating the number agents who enter the pool in the interval $[0,t_0]$. By Bayes rule,
$$ \P{\tilde{Z}_{t_0}=\ell} = \sum_{k=0}^\infty \P{\tilde{Z}_{t_0}=\ell, K=k} = \sum_{k=0}^\infty \P{\tilde{Z}_{t_0}=\ell | K=k} \cdot \frac{(mt_0)^k e^{-mt_0}}{k!},$$
where the last equation follows by the fact that arrival rate of the agents is a Poisson random variable of rate $m$.
In \autoref{fact:expectedpoolsize} we show that condition on an agent $a$ arrives in the interval $[0,t_0]$, the probability that he is in the pool at time $t_0$ is $(1-e^{-t_0})/t_0$.
Therefore, conditioned on $K=k$, the distribution of the number of agents at time $t_0$
is a Binomial random variable $B(k,p)$, where $p:=(1-e^{-t_0})/t_0)$.
So,
\begin{eqnarray*} \P{\tilde{Z}_{t_0}=\ell} &=& \sum_{k=\ell}^\infty {k\choose \ell} \cdot p^\ell \cdot (1-p)^{k-\ell}  \frac{(mt_0)^k e^{-mt_0}}{k!} \\
&= & \sum_{k=\ell}^\infty \frac{m^k e^{-mt_0}}{\ell! (k-\ell)!} (1-e^{-t_0})^\ell (t_0-1+e^{-t_0})^{k-\ell} \\
&\leq & \frac{m^\ell e^{-mt_0}}{\ell!} \sum_{k=\ell}^\infty \frac{(mt_0)^{k-\ell}}{(k-\ell)!} = \frac{m^\ell}{\ell!}.
\end{eqnarray*}
\end{proof}
This completes the proof of \autoref{lem:algGreedyEZ}.
\end{proof}

So, in the rest of the proof we just need to upper bound $\EE{Z\sim \pi}{Z}$.
Unfortunately, we do not have any closed form expression of the stationary distribution, $\pi(.)$.
Instead, we use the balance equations of the Markov Chain defined on $Z_t$ to characterize $\pi(.)$ and upper bound $\EE{Z\sim \pi}{Z}$.

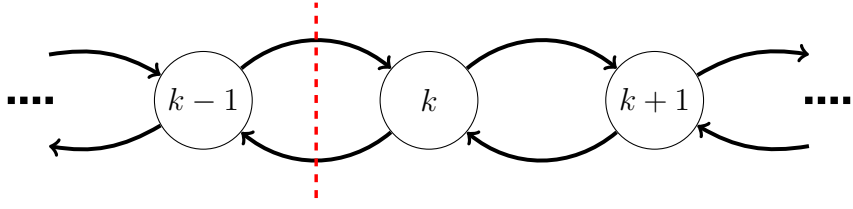
\begin{figure}[t]
\centering
\begin{tikzpicture}[inner sep=0, minimum size=13mm]
\tikzstyle{every node}=[circle,draw];
\node at (0,0) (b) {$k$};
\node at (-3,0) (a) {$k-1$};
\node at (3,0) (c) {$k+1$};
\node at  (-5.7,-.5) [draw=none] (d) {};
\node at  (5.7,.5) [draw=none] (e) {};
\node at  (-5.7,.5) [draw=none] (d1) {};
\node at  (5.7,-.5) [draw=none] (e1) {};
\draw[->,line width=1.5pt]  (a) [bend left=40] edge (b) (a) [bend left=20] edge (d)
  (d1) [bend left=20] edge (a) (e1) [bend left=20] edge (c) (c) [bend left=20] edge (e)
  (b) [bend left=40] edge (a)
  (b) [bend left=40] edge (c)
  (c) [bend left=40] edge (b);
\draw [dotted, line width=2.8pt] (a)+(-2,0) -- ++(-2.6,0);
\draw [dotted, line width=2.8pt] (c)+(2,0) -- ++(2.6,0);
\draw [dashed,line width=1.3,color=red] (b)+(-1.5,-1.3) -- ++(-1.5,1.3);
\end{tikzpicture}
\caption{An illustration of the transition paths of the $Z_t$ Markov Chain under the Greedy algorithm}
\label{fig:Greedybalance}
\end{figure}

Let us rigorously define the transition probability operator of the Markov Chain on $Z_t$.
For any pool size $k$, the Markov Chain transits only to the states $k+1$ or $k-1$. It transits to state $k+1$ if a new agent arrives and the market-maker cannot match him (i.e., the new agent does not have any edge to the agents currently in the pool) and the Markov Chain transits to the state $k-1$ if a new agent arrives and is matched or an agent currently in the pool gets critical.
Thus, the transition rates $r_{k\to k+1}$ and $r_{k\to k-1}$ are defined as follows,
\begin{eqnarray}
	r_{k \to k+1} &:=&
	m\Big(1-\frac{d}{m}\Big)^k  \label{eq:transGreedy1}\\ 
	r_{k\to k-1} &:=& 
	k + m\Big(1-\Big(1-\frac{d}{m}\Big)^k\Big).
	\label{eq:transGreedy2}
\end{eqnarray}
In the above equations we used the fact that agents arrive at rate $m$, they perish at rate 1 and the probability of an acceptable transaction between two agents is $d/m$.

Let us write down the balance equation for the above Markov Chain (see equation \eqref{eq:balanceequations} for the full generality).
Consider the cut separating the states $0,1,2,\ldots,k-1$ from the rest (see \autoref{fig:Greedybalance} for an illustration). It follows that
\begin{equation}
\label{eq:balanceGreedy}
\pi(k-1) r_{k-1\to k} = \pi(k) r_{k\to k-1}
\end{equation}

Now, we are ready to characterize the stationary distribution $\pi(.)$.
In the following proposition we show that there is a number $z^* \leq \log(2)m/d$ such that under the stationary distribution,
the size of the pool is highly concentrated in an interval  of length $O(\sqrt{m/d})$ around $z^*$.
\begin{proposition}\label{prop:Greedyconcentration}
There exists $m/(2d+1)\leq k^* < \log(2)m/d$ 
such that for any $\sigma>1$,
$$ \PP{\pi}{k^*- \sigma\sqrt{2m/d} \leq Z \leq k^*+ \sigma\sqrt{2m/d}} \geq 1-O(\sqrt{m/d})e^{-\sigma^2}.$$
\end{proposition}
\begin{proof}
Let us define $f:\R\to\R$ as an interpolation of the difference of transition rates over the reals,
$$ f(x):=m(1-d/m)^{x} - (x+m(1-(1-d/m)^x)).$$
In particular, observe that $f(k)=r_{k\to k+1} - r_{k\to k-1}$.
The above function is a decreasing convex function over non-negative reals. We define $k^*$
as the unique root of this function.
Let $k^*_{\min}:=m/(2d+1)$ and $k^*_{\max}:=\log(2)m/d$. We show that $f(k^*_{\min})\geq 0$ and
$f(k^*_{\max}) \leq 0$. This shows that $k^*_{\min} \leq k^*< k^*_{\max}$. 
\begin{eqnarray*}
f(k^*_{\min}) &\geq & -k^*_{\min} - m +2m(1-d/m)^{k^*_{\min}} \geq 2m\Big(1-\frac{k^*_{\min} d}{m}\Big) - k^*_{\min} = 0,\\
 f(k^*_{\max}) &\leq & -k^*_{\max} - m +2m(1-d/m)^{k^*_{\max}} \leq -k^*_{\max} - m + 2m e^{-(k^*_{\max}) d/m}  = -k^*_{\max} \leq 0.
 \end{eqnarray*}
In the first inequality we used equation \eqref{eq:bernoulli}.

It remains to show that $\pi$ is highly concentrated around $k^*$.
We prove this in several steps

\begin{claim}\label{cl:Greedy}
For any integer $k\geq k^*$
$$ \frac{\pi(k+1)}{\pi(k)} \leq e^{-(k-k^*)d/m}.$$
And, for any $k\leq k^*$, $\pi(k-1)/\pi(k) \leq e^{-(k^*-k+1)d/m}$.
\end{claim}
\begin{proof}
For $k\geq k^*$, by \eqref{eq:transGreedy1}, \eqref{eq:transGreedy2}, \eqref{eq:balanceGreedy},
$$ \frac{\pi(k)}{\pi(k+1)} = \frac{(k+1) + m(1-(1-d/m)^{k+1})}{m(1-d/m)^{k}} = \frac{k-k^*+1-m(1-d/m)^{k+1} +2m(1-d/m)^{k^*} }{m(1-d/m)^{k}}$$
where we used the definition of $k^*$.
Therefore,
\begin{eqnarray*}
\frac{\pi(k)}{\pi(k+1)} \geq -(1-d/m) + \frac{2}{(1-d/m)^{k-k^*}} \geq \frac{1}{(1-d/m)^{k-k^*}} \geq e^{-(k-k^*)d/m} 
\end{eqnarray*}
where the last inequality uses $1-x \leq e^{-x}$.  Multiplying across the inequality yields the claim.
Similarly, we can prove the second conclusion. For $k\leq k^*$,
\begin{eqnarray*}
\frac{\pi(k-1)}{\pi(k)} &=& \frac{k-k^*-m(1-d/m)^{k}+2m(1-d/m)^{k^*}}{m(1-d/m)^{k-1}}\\
&\leq & -(1-d/m)+2(1-d/m)^{k^*-k+1} \leq (1-d/m)^{k^*-k+1} \leq e^{-(k^*-k+1)d/m},
\end{eqnarray*}
where the second to last inequality uses $k\leq k^*$.
\end{proof}

By repeated application of the above claim, for any integer $k\geq k^*$, we get\footnote{$\lceil k^*\rceil$ indicates the smallest integer larger than $k^*$.}

\begin{equation}
\label{eq:Greedyprobupper}
 \pi(k) \leq \frac{\pi(k)}{\pi(\lceil k^*\rceil)} \leq \exp\Big(-\frac{d}{m}\sum_{i=\lceil k^* \rceil}^{k-1} (i-k^*)\Big) \leq \exp(-d(k-k^*-1)^2/2m).
 \end{equation}

We are almost done.

For any $\sigma>0$,
\begin{eqnarray*}
\sum_{k=k^*+1+\sigma\sqrt{2m/d}}^\infty \pi(k) 
\leq  \sum_{k=k^*+1+\sigma\sqrt{2m/d}}^\infty e^{-d(k-k^*-1)^2/2m} &=& \sum_{k=0}^{\infty} e^{-d(k+\sigma \sqrt{2m/d})^2/2m}\\
&\leq & \frac{e^{-\sigma^2}}{\min\{1/2,\sigma\sqrt{d/2m}\}}
\end{eqnarray*}
The last inequality uses equation \eqref{eq:sumsqexp}.
We can similarly upper bound $\sum_{k=0}^{k^*-\sigma\sqrt{2m/d}} \pi(k)$.
\end{proof}

The above proposition shows that the probability that the size of the pool falls outside an interval of length $O(\sqrt{m/d})$ around $k^*$ drops exponentially fast as the market  size grows.
We also remark that the upper bound on $k^*$ becomes tight as $d$ goes to infinity.

\begin{lemma}\label{prop:Greedyexpectation}
For $k^*$ as in \autoref{prop:Greedyconcentration} ,
$$\EE{Z\sim \pi}{Z} \leq k^*+O(\sqrt{m/d}\log(m/d)). $$
\end{lemma}
\begin{proof}
Let 
Let $\Delta\geq 0$ be a parameter that we fix later.
\begin{equation}
\label{eq:Greedyexpfirst}
 \EE{Z\sim \pi}{Z} \leq k^*+\Delta+ \sum_{i=k^* + \Delta+1}^\infty i\pi(i).
 \end{equation}
By 
equation \eqref{eq:Greedyprobupper},
\begin{eqnarray}
\sum_{i=k^*+\Delta+1}^\infty i\pi(i) 
&=& \sum_{i=\Delta+1}^\infty e^{-d(i-1)^2/2m} (i+k^*)\nonumber \\
&=& \sum_{i=\Delta}^\infty e^{-d i^2/2m}(i-1) +\sum_{i=\Delta}^\infty e^{-di^2/2m}(k^*+2)\nonumber \\
&\leq & \frac{e^{-d(\Delta-1)^2/2m}}{d/m} + (k^*+2)\frac{e^{-d\Delta^2/2m}}{\min\{1/2,d\Delta/2m\}},
 \label{eq:Greedyexplast}
\end{eqnarray}
where we used equations \eqref{eq:sumsqexp} and \eqref{eq:sumexpexp}.
Letting $\Delta:=1+2\sqrt{m/d}\log(m/d)$ in the above equation, the right hand side is at most $1$. The lemma follows from \eqref{eq:Greedyexpfirst} 
and the above equation.
\end{proof}

Now, \autoref{thm:Greedyanal}.
 follows immediately  by \autoref{lem:algGreedyEZ} and \autoref{prop:Greedyexpectation} because we have,

$$ \loss{\algGreedy} \le  \frac{1}{m}(k^* + O(\sqrt{m}\log m)) \le \frac{\log(2)}{d} + o(1) $$

\subsection{Loss of the Patient Algorithm}
\label{sec:critical}
%

Throughout this section we use $Z_t$ to denote  the size of the pool under $\algcritical$.  Let $\pi: \mathbb{N} \to \mathbb{R}_+$ be the unique stationary distribution of the Markov Chain on $Z_t$, and let $\poolsize:=\EE{Z\sim\pi}{Z}$ be the expected size of the pool under that distribution.

 By \autoref{Z_suff_stat}, at any point in time $G_t$ is an Erd\"os-R\'eyni random graph. So once an agent becomes critical, he has at least one acceptable transaction with probability $1 - (1-d/m)^{Z_t-1}$. Since each agent becomes critical with rate $1$, if we run $\algcritical$ for a sufficiently long time, then $\loss{\algcritical}\approx \poolsize(1-d/m)^{\poolsize-1}$. The following lemma makes the above discussion rigorous.


\begin{lemma}\label{lem:criticallemma}
For any $\epsilon > 0$ and $T > 0$,
$$ \loss{\algcritical} \leq \frac{1}{m}\EE{Z\sim\pi}{Z(1-d/m)^{Z-1}} + \frac{\tau_{\mix}(\eps)}{T} + \frac{\eps m}{d^2}. $$
\end{lemma}

\begin{proof}
By linearity of expectation,
\begin{eqnarray*}
\loss{\algcritical} = \frac{1}{m\cdot T} \E{\int_{t=0}^T Z_t(1 - d/m)^{Z_t-1} dt} &=& \frac{1}{m\cdot T}\int_{t=0}^T \E{Z_t(1 - d/m)^{Z_t-1}} dt.
\end{eqnarray*}
Since for any $t\geq 0$, $\E{Z_t (1-d/m)^{Z_t-1}} \leq \E{Z_t}\leq \E{\tilde{Z}_t} \leq m$,
we can write
\begin{eqnarray*}
\loss{\algcritical} &\leq& \frac{\tau_{\mix}(\eps)}{T} + \frac{1}{m\cdot T}\int_{t=\tau_{\mix}(\eps)}^{T}\sum_{i=0}^\infty (\pi(i)+\eps) i(1 - d/m)^{i-1} dt  \\
&\leq & \frac{\tau_{\mix}(\eps)}{T} + \frac{\EE{Z\sim\pi}{Z(1-d/m)^{Z-1}}}{m} + \frac{\eps m}{d^2}
\end{eqnarray*}
where the last inequality uses the identity
$ \sum_{i=0}^\infty i(1-d/m)^{i-1} = m^2/d^2.$
\end{proof}

So in the rest of the proof we just need to lower bound $\EE{Z\sim\pi}{Z(1-d/m)^{Z-1}}$. Similar to the Greedy case, we do not have a closed form expression for the stationary distribution, $\pi(.)$. Instead, we use the balance equations of the Markov Chain on $Z_t$ to show that
$\pi$ is highly concentrated around a number $k^*$ where $k^*\in [m/2,m]$.

Let us start by defining the transition probability operator of the Markov Chain on $Z_t$. For any pool size $k$, the Markov Chain transits only to states $k+1$, $k-1$, or $k-2$. The Markov Chain transits to state $k+1$ if a new agent arrives, to the state $k-1$ if an agent gets critical and the the Planner cannot match him, and it transits to state $k-2$ if an agent gets critical and  the Planner  matches him.

Remember that agents arrive with the rate $m$, they become critical with the rate of 1 and the probability of an acceptable transaction between two agents is $d/m$. Thus, the transition rates $r_{k\to k+1}$, $r_{k\to k-1}$, and $r_{k\to k-2}$ are defined as follows,
\begin{eqnarray}
	\label{eq:transcritical1}
	r_{k \to k+1} &:=&
	m\\ 
	r_{k\to k-1} &:=& 
	k\Big(1-\frac{d}{m}\Big)^{k-1}
	\label{eq:transcritical2} \\
	r_{k\to k-2} &:=&
	k\Big(1- \Big(1-\frac{d}{m}\Big)^{k-1}\Big).
	\label{eq:transcritical3}
\end{eqnarray}

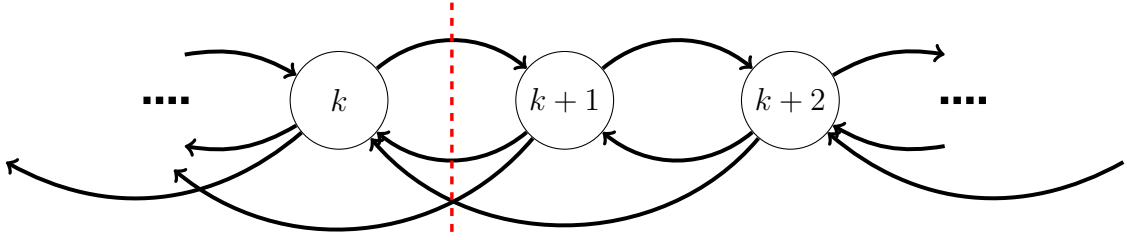
\begin{figure}[t]
\centering
\begin{tikzpicture}[inner sep=0, minimum size=13mm]
\tikzstyle{every node}=[circle,draw];
\node at (0,0) (b) {$k+1$};
\node at (-3,0) (a) {$k$};
\node at (3,0) (c) {$k+2$};
\node at  (-5.7,-.5) [draw=none] (d) {};
\node at  (5.7,.5) [draw=none] (e) {};
\node at  (-5.7,.5) [draw=none] (d1) {};
\node at  (5.7,-.5) [draw=none] (e1) {};
\node at  (-8,-.5) [draw=none] (d2) {};
\node at  (8,-.5) [draw=none] (e2) {};

\draw[->,line width=1.5pt]  (a) [bend left=40] edge (b) (a) [bend left=20] edge (d)
  (d1) [bend left=20] edge (a) (e1) [bend left=20] edge (c) (c) [bend left=20] edge (e)
  (b) [bend left=40] edge (a)
  (b) [bend left=40] edge (c)
  (c) [bend left=40] edge (b)
  (c) [bend left=50] edge (a)
  (a) [bend left=35] edge (d2)
  (e2) [bend left=35] edge (c)
  (b) [bend left=45] edge (d);

\draw [dotted, line width=2.8pt] (a)+(-2,0) -- ++(-2.6,0);
\draw [dotted, line width=2.8pt] (c)+(2,0) -- ++(2.6,0);
\draw [dashed,line width=1.3,color=red] (b)+(-1.5,-1.75) -- ++(-1.5,1.3);
\end{tikzpicture}
\caption{An illustration of the transition paths of the $Z_t$ Markov Chain under the Patient Algorithm}
\label{fig:criticalbalance}
\end{figure}

Let us write down the balance equation for the above Markov Chain (see equation \eqref{eq:balanceequations} for the full generality).
Consider the cut separating the states $0,1,2,\ldots,k$ from the rest (see \autoref{fig:criticalbalance} for an illustration). It follows that
\begin{equation}
\label{eq:balancecritical}
\pi(k) r_{k\to k+1} = \pi(k+1)r_{k+1\to k} + \pi(k+1)r_{k+1\to k-1} + \pi(k+2)r_{k+2\to k}
\end{equation}

Now we can characterize $\pi(.)$. We show that 
under the stationary distribution, the size of the pool is highly concentrated around a number $k^*$ where  $k^*\in [m/2,m]$. Remember that under the Greedy algorithm, the concentration was around $k^* \in [\frac{m}{2d+1}, \frac{\log(2)m}{d}]$, whereas here it is at least $m/2$.

\begin{proposition}[Patient Concentration]\label{prop:criticalconcentration}
There exists a number $m/2-2 \leq k^* \leq m-1$ such that for any $\sigma\geq 1$,
$$ \PP{\pi}{k^*- \sigma\sqrt{4m} \leq Z} \geq 1-2\sqrt{m} e^{-\sigma^2}, \text{  } \P{Z\leq  k^*+ \sigma\sqrt{4m}} \geq 1-8\sqrt{m} e^{-\frac{\sigma^2\sqrt{m}}{2\sigma+\sqrt{m}}}.$$
\end{proposition}

\begin{proof}
The proof idea is similar to \autoref{prop:Greedyconcentration}.
First,  let us rewrite \eqref{eq:balancecritical} by replacing transition probabilities from \eqref{eq:transcritical1}, \eqref{eq:transcritical2}, and \eqref{eq:transcritical3}:
\begin{equation}\label{eq:balancecritical2}
m\pi(k) = 	(k+1) \pi(k+1) + (k+2)\Big(1 - \Big(1-\frac{d}{m}\Big)^{k+1} \Big) \pi(k+2)
\end{equation}
Let us define a continous $f:\R\to\R$ as follows,
\begin{equation}
\label{eq:rootcritical}
 f(x) := m - (x+1) - (x+2)(1-(1-d/m)^{x+1}).
 \end{equation}
It follows that 
$$f(m-1)\leq 0,  f(m/2-2)  >0,$$
so $f(.)$ has a root $k^*$ such that $m/2-2 < k^* < m$.
In the rest of the proof we show that the states that are far from $k^*$ have very small probability in the stationary distribution. 


\begin{claim}
For any integer $k\leq k^*$, 
$$\frac{\pi(k)}{\max\{\pi(k+1),\pi(k+2)\}} \leq e^{-(k^*-k)/m}.$$
Similarly, for any integer $k\geq k^*$, $\frac{\min\{\pi(k+1),\pi(k+2)\}}{\pi(k)} \leq e^{(k-k^*)/(m+k-k^*)}$.
\end{claim}
\begin{proof}
For $k\leq k^*$, by equation \eqref{eq:balancecritical2},
\begin{eqnarray*}
 \frac{\pi(k)}{\max\{\pi(k+1),\pi(k+2)\}}
& \leq& \frac{(k+1) + (k+2) (1-(1-d/m)^{k+1})}{m} \\
& \leq & \frac{(k-k^*) + (k^*+1) + (k^*+2)(1-(1-d/m)^{k^*+1})}{m}\\
&=&  1 - \frac{k^*-k}{m} \leq e^{-(k^*-k)/m},
\end{eqnarray*}
where the last equality follows by the definition of $k^*$
and the last inequality uses $1-x\leq e^{-x}$.
The second conclusion can be proved similarly. 
For $k\geq k^*$,
\begin{eqnarray*}
\frac{\min\{\pi(k+1),\pi(k+2)\}}{\pi(k)} &\leq& \frac{m}{(k+1)+(k+2)(1-(1-d/m)^{k+1})}\\
&\leq & \frac{m}{(k-k^*) + (k^*+1)+(k^*+2)(1-(1-d/m)^{k^*+1})}\\
&= & \frac{m}{m+k-k^*} = 1-\frac{k-k^*}{m+k-k^*} \leq e^{-(k-k^*)/(m+k-k^*)}.
\end{eqnarray*}
where the equality follows by the definition of $k^*$. 
\end{proof}
Now, we use the above claim to upper-bound $\pi(k)$ for values $k$ that are far from $k^*$.
First, fix $k\leq k^*$. Let $n_0,n_1,\ldots$ be sequence of integers defined as follows:  $n_0=k$, and $n_{i+1} := \argmax\{\pi(n_i+1), \pi(n_i+2)\}$ for $i\geq 1$. It follows that,
\begin{eqnarray}
\pi(k) \leq \prod_{i: n_i \leq k^*} \frac{\pi(n_i)}{\pi(n_{i+1})} \leq   \exp\Big( -\sum_{i:n_i \leq k^*} \frac{k^*-n_i}{m}\Big)  
\leq  \exp\Big(-\sum_{i=0}^{(k^*-k)/2} \frac{2i}{m}\Big)
\leq  e^{-(k^*-k)^2/4m},
\label{eq:Greedyleftsidecon}
\end{eqnarray}
where the second to last inequality uses $|n_i-n_{i-1}| \leq 2$.

Now, fix $k\geq k^*+2$. In this case we construct the following sequence of integers,
$n_0=\lfloor k^*+2\rfloor$, and $n_{i+1}:=\argmin\{\pi(n_i+1),\pi(n_i+2)\}$ for $i\geq 1$. 
Let $n_j$ be the largest number in the sequence that is at most $k$ (observe that $n_j=k-1$ or $n_j=k$). We upper-bound $\pi(k)$ by upper-bounding $\pi(n_j)$,
\begin{align}
\pi(k)\leq \frac{m\cdot \pi(n_j)}{k} \leq 2\prod_{i=0}^{j-1} \frac{\pi(n_i)}{\pi(n_{i+1})} 
&\leq 2\exp\Big(-\sum_{i=0}^{j-1} \frac{n_i-k^*}{m+n_i-k^*}\Big) \nonumber \\
&\leq 2\exp\Big(-\sum_{i=0}^{(j-1)/2} \frac{2i}{m+k-k^*} \Big)
\leq 2\exp\Big(\frac{-(k-k^*-1)^2}{4(m+k-k^*)}\Big).
\label{eq:Greedyrightsidecon}
\end{align}
To see the first inequality note that if $n_j=k$, then there is nothing to show; otherwise we have $n_j=k-1$. In this case by equation \eqref{eq:balancecritical2}, $m\pi(k-1)\geq k\pi(k)$.
The second to last inequality uses the fact that $|n_i-n_{i+1}|\leq 2$. 

We are almost done. The proposition follows from \eqref{eq:Greedyrightsidecon} and \eqref{eq:Greedyleftsidecon}. 
First, for $\sigma\geq 1$, let $\Delta=\sigma\sqrt{4m}$, then by equation \eqref{eq:sumsqexp}
\begin{eqnarray*}
 \sum_{i=0}^{k^*-\Delta} \pi(i) \leq
 \sum_{i=\Delta}^\infty e^{-i^2/4m}
\leq \frac{e^{-\Delta^2/4m}}{\min\{1/2,\Delta/4m\}} \leq 2\sqrt{m} e^{-\sigma^2}.
\end{eqnarray*}
Similarly, 
\begin{eqnarray*}
\sum_{i=k^*+\Delta}^\infty \pi(i) \leq 2\sum_{i=\Delta+1}^\infty e^{-(i-1)^2/4(i+m)}  
&\leq& 2\sum_{i=\Delta}^\infty e^{-i/(4+\sqrt{4m}/\sigma)}\\
&\leq& 2\frac{e^{-\Delta/(4+\sqrt{4m}/\sigma)}}{1-e^{-1/(4+\sqrt{4m})}} \leq 8\sqrt{m}e^{\frac{-\sigma^2\sqrt{m}}{2\sigma+\sqrt{m}}}
\end{eqnarray*}
This completes the proof of \autoref{prop:criticalconcentration}.
\end{proof}

\begin{lemma}\label{prop:criticalexpectation}
For any $d\geq 0$ and sufficiently large $m$, 
$$\EE{Z\sim \pi}{Z(1-d/m)^Z} \leq \max_{z\in [m/2,m]} (z+\tilde{O}(\sqrt{m}))(1-d/m)^z + 2.$$
\end{lemma}

\begin{proof}
Let $\Delta:=3\sqrt{m}\log(m)$,
and let $\beta:=\max_{z\in [m/2-\Delta,m+\Delta]} z(1-d/m)^z$.
\begin{eqnarray}
\EE{Z\sim \pi}{Z(1-d/m)^Z} \leq \beta + \sum_{i=0}^{m/2-\Delta-1} \frac{m}{2}\pi(i) (1-d/m)^i + \sum_{i=m+\Delta}^{\infty} i\pi(i)  (1-d/m)^m
\label{eq:expcriticalallterms}
\end{eqnarray}
We upper bound each of the terms in the right hand side separately. We start with upper bounding $\beta$. Let $\Delta':=4(\log(2m)+1)\Delta$.
\begin{eqnarray}
 \beta &\leq& \max_{z\in [m/2,m]} z(1-d/m)^z + m/2 (1-d/m)^{m/2} ((1-d/m)^{-\Delta} - 1) + (1-d/m)^m \Delta\nonumber \\
 &\leq & \max_{z\in [m/2,m]} (z+\Delta'+\Delta) (1-d/m)^z + 1.
 \label{eq:expcriticalfirst}
 \end{eqnarray}
To see the last inequality we consider two cases. If $(1-d/m)^{-\Delta} \leq 1+\Delta'/m$ then the inequality obviously holds. Otherwise, (assuming $\Delta'\leq m$),
$$(1-d/m)^\Delta \leq \frac{1}{1+\Delta'/m} \leq 1-\Delta'/2m,$$
By the definition of $\beta$,
$$ \beta \leq (m+\Delta) (1-d/m)^{m/2-\Delta} \leq 2m(1-\Delta'/2m)^{m/2\Delta-1} \leq 2me^{\Delta'/4\Delta-1} \leq 1.$$


It remains to upper bound the second and the third term in \eqref{eq:expcriticalallterms}.
We start with the second term. By \autoref{prop:criticalconcentration},
\begin{eqnarray}
 \sum_{i=0}^{m/2-\Delta-1} \pi(i) 
 \leq \frac{1}{m^{3/2}}.
\label{eq:expcriticalsecond}
\end{eqnarray}
where we used equation \eqref{eq:sumsqexp}.
On the other hand, by equation \eqref{prop:criticalconcentration},
\begin{eqnarray}
\sum_{i=m+\Delta}^\infty i\pi(i) \leq 
 e^{-\Delta/(2+\sqrt{m})}(\frac{m}{1-e^{-1/(2+\sqrt{m})}} + \frac{ 2\Delta+4}{1/(2+\sqrt{m})^2}) \leq \frac{1}{\sqrt{m}}.
\label{eq:expcriticalthird}
\end{eqnarray}
where we used equation \eqref{eq:sumexp}.

The lemma follows from \eqref{eq:expcriticalallterms}, \eqref{eq:expcriticalfirst}, \eqref{eq:expcriticalsecond} and \eqref{eq:expcriticalthird}.
\end{proof}

\noindent Now the proof of  \autoref{thm:criticalanal} follows immediately by combining \autoref{lem:criticallemma} and \autoref{prop:criticalexpectation}.

\subsection{Loss of the $\apatient{\alpha}$ Algorithm}\label{sec:alphaanal}


Our idea is to \emph{slow down} the process and use \autoref{thm:criticalanal} to analyze the $\apatient{\alpha}$ algorithm. More precisely, instead of analyzing $\apatient{\alpha}$ algorithm on a $(m,d,1)$ market we analyze the Patient algorithm on a $(m/\bar\alpha,d/\bar\alpha,1)$ market. First we need to prove a lemma on the equivalence of markets with different criticality rates. 

\begin{definition}[Market Equivalence]\label{def:equivalence}
An \emph{$\alpha$-scaling} of a dynamic matching market $(m,d,\lambda)$ is defined as follows.
Given any realization of this market, i.e., given $A^c_t,A^n_t,E$ for any $0\leq t\leq \infty$, we construct another realization with the same set of acceptable transactions and the sets $A^c_{\alpha\cdot t},A^n_{\alpha \cdot t}$. 
We say two dynamic matching markets $(m,d,\lambda)$ and $(m',d',\lambda')$
are {\em equivalent} if one is an $\alpha$-scaling of the other. 
%
\end{definition}

It turns out that for any $\alpha\geq 0$, and any time $t$, any of the Greedy, Patient or $\apatient{\alpha}$ algorithms (and in general any time-scale independent online algorithm) the set of matched agents at time $t$ of a realization of a $(m,d,\lambda)$ matching market is the same as the set of matched agents at time $\alpha t$ of  an $\alpha$-scaling of that realization. 
The following fact makes this rigorous.

\begin{proposition}\label{lem:equivalence}
For any $m,d,\lambda$ the $(m/\lambda,d/\lambda,1)$ matching market is equivalent to the $(m,d,\lambda)$ matching market. 
\end{proposition}
\begin{proof}
It turns out that a $1/\lambda$-scaling of the $(m,d,\lambda)$-market is exactly the same as the $(m/\lambda, d/\lambda,1)$ market. 
%
In particular,  the arrival rate in the scaled market is exactly $m/\lambda$, the criticality rate is 1 and the probability of an acceptable transaction between two agents in the pool remains unchanged, i.e., $\frac{d}{m} = \frac{d/\lambda}{m/\lambda}$.
\end{proof}

Now,  \autoref{thm:alphaanal} follows simply by combining the above fact and \autoref{thm:criticalanal}.
First, by the additivity of the Poisson process, the loss of the $\apatient{\alpha}$ algorithm in a $(m,d,1)$ matching market
is equal to the loss of the Patient algorithm in a $(m,d,\bar\alpha)$ matching market, where $\bar\alpha=1/\alpha+1$. 
Second, by the above fact, the loss of the Patient algorithm
on a $(m,d,\bar\alpha)$ matching market at time $T$ is the same as the loss of this algorithm on a $(m/\bar\alpha,d/\bar\alpha,1)$ market at time $\bar\alpha T$. The latter is upper-bounded in \autoref{thm:criticalanal}.

\section{Welfare Analysis}\label{sec:welfare}

\begin{theorem}
\label{thm:welfarecritical}
For sufficiently large $m$, any $T\geq 0$, $\delta \geq 0$ and $\eps < 1/2m^2$,
$$ \W{\algcritical} \geq \frac{2(T- \tau_{\mix}(\eps))}{T(\delta+2)}(1-e^{-d/2}(1+\tilde{O}(1/\sqrt{m}))) - \frac{2}{T(\delta+2)^2} -\frac{3}{m}$$
where $\tau_{\mix}(\eps) \leq 8\log(m)\log(4/\eps)$.
As a corollary, for any $\alpha \geq 0$, and $\bar\alpha=1/\alpha+1$,
$$ \W{\apatient{\alpha}} \geq \frac{2(\bar\alpha T -\tau_{\mix}(\eps))}{\bar\alpha T(\delta/\bar\alpha+2)} (1-e^{-d/2\bar\alpha} (1+\tilde{O}(\sqrt{\bar\alpha/m}))) - \frac{2}{\bar\alpha T(\delta+2)^2} - \frac{3\bar\alpha}{m}.$$ 
\end{theorem}
\begin{theorem}
If $m > 10d$, for any $T\geq 0$,
$$ \W{\algGreedy} \leq 1-\frac{1}{2d+1+d^2/m}.$$
\end{theorem}


Say an agent $a$ is arrived at time $t_a(a)$. We let $X_t$ be the sum of the potential utility of the agents in $A_t$:
$$X_t = \sum_{a \in A_t}e^{-\delta (t-t_a(a))}, $$
i.e., if we match all of the agents currently in the pool immediately, the total utility that they receive is exactly $X_t$.

For $t_0,\eps>0$, let $W_{t_0, t_0 + \epsilon}$
be the expected total utility of  the agents who are matched in the interval $[t_0, t_0+\eps]$.
By definition the social welfare of an online algorithm, we have:
$$\W{\algcritical} = \E{\frac{1}{T}\int_{t=0}^{T} W_{t, t+dt} dt} = \frac{1}{T}\int_{t=0}^{T} \E{W_{t, t+dt}} dt $$

\subsection{Welfare of the Patient Algorithm}

All agents are equally likely to become critical at each moment. From the perspective of the planner, all agents are equally likely to be the neighbor of a critical agent. Hence, the \emph{expected} utility of each of the agents who are matched at time $t$ under the Patient algorithm is $X_t / Z_t$. Thus,

\begin{align}
\W{\algcritical} = \frac{1}{mT}\int_{t=0}^T  \E{2\frac{X_t}{Z_t} Z_t (1-(1-d/m)^{Z_t})dt}  = \frac{2}{mT}\int_{t=0}^T \E{ X_t (1-(1-d/m)^{Z_t})}dt
\end{align}

First, we prove the following claim.
\begin{lemma}
For any $\eps < 1/2m^2$, and $t\geq \tau_{\mix}(\eps)$,
$$ \E{X_t(1-(1-d/m)^{Z_t})} \geq \E{X_t} \Big(1-\frac{3}{m} - e^{-d/2}(1+\frac{4\log^2 m}{\sqrt{m}})\Big) - \frac{1}{\sqrt{m}}
$$
\end{lemma}
\begin{proof}
Let $\beta:=m/2-\sigma\sqrt{4m}$ for $\sigma\geq 1$ that we fix later.
First, observe that
\begin{eqnarray}
 \E{X_t(1-(1-d/m)^{Z_t})} &\geq& \E{X_t(1-(1-d/m)^{Z_t} | Z_t \geq \beta} \cdot \P{Z_t\geq \beta} \nonumber \\
 &\geq & \E{X_t | Z_t\geq \beta} \P{Z_t \geq \beta} (1-(1-d/m)^\beta)
 \end{eqnarray}
 
 On the other hand,
 \begin{eqnarray}
 \E{X_t | Z_t \geq \beta} \geq \E{X_t} - \E{X_t | Z_t < \beta} \P{Z_t < \beta} \geq \E{X_t}-\beta\cdot \P{Z_t<\beta}
 \end{eqnarray}
where we used that conditioned on $Z_t\leq \beta$ we have $X_t\leq \beta$ with probability 1.
Let $\sigma =\sqrt{2\log(2m)}$. Since $t\geq \tau_{\mix}(\eps)$, using \autoref{prop:criticalconcentration},
$$\P{Z_t\geq \beta} \geq 1-\sqrt{m}e^{-\sigma^2}-\eps \geq 1-1/m^{-3/2}.$$
Putting above together, we get
\begin{eqnarray*} \E{X_t(1-(1-d/m)^{Z_t}} &\geq& (\E{X_t} - m^{-1/2}) (1-m^{-3/2}) \Big(1-\frac{(1-d/m)^{m/2}}{(1-d/m)^{\sigma\sqrt{4m}}}\Big) \\
&\geq &
\E{X_t} \Big(1-\frac{3}{m} - e^{-d/2}(1+\frac{4\log^2 m}{\sqrt{m}})\Big) - \frac{1}{\sqrt{m}}
\end{eqnarray*}
To see the last equation we need to consider two cases. Let $\Delta:=4\log(m)\sigma\sqrt{4m}$.
Now, if $(1-\frac{d}{m})^{-\sigma\sqrt{4m}} \leq 1+\Delta/m$, then
$$1-(1-d/m)^{m/2-\sigma\sqrt{4m}} \geq 1-e^{-d/2} (1+\Delta/m).$$
Otherwise, for sufficiently large $m$, $(1-d/m)^{\sigma\sqrt{4m}} \leq 1-\Delta/2m$. So,
$$ 1-(1-d/m)^{m/2-\sigma\sqrt{4m}} \geq 1-(1-\Delta/2m)^{\frac{m}{2\sigma\sqrt{4m}} -1} \geq 1-e^{-\log(m)+\Delta/2m} \leq 1-2/m.$$

\end{proof}

Let $\eps < 1/2m^2$.  By above lemma, we have:
\begin{eqnarray}
\W{\algcritical} 
\geq \frac{2}{mT} \int_{t=\tau_{\mix}(\eps)}^T \Big(\E{X_t} \Big(1-\frac{3}{m} - e^{-d/2}(1+\frac{4\log^2 m}{\sqrt{m}})\Big) - \frac{1}{\sqrt{m}}\Big)
\end{eqnarray}

It remains to lower-bound $\E{X_t}$. This is done in the following lemma.
%

\begin{lemma}\label{totalutil.patient}
For any $t_1 \geq 0$, 
$$\E{X_{t_1}} \geq \frac{m}{\delta+2} (1-e^{-(\delta+2)t_1}) $$
\end{lemma}
\begin{proof}
Let $\eps > 0$ and be very close to zero (eventually we let $\eps\to0$). Since we have a $(m,d,1)$ matching market, using equation \eqref{eq:dtzt} for any $t\geq 0$ we have,
\begin{eqnarray*}
\E{X_{t + \epsilon} | X_t, Z_t} = X_t(e^{-\epsilon\delta}) + m\epsilon - \epsilon Z_t \Big(\frac{X_t}{Z_t}(1-d/m)^{Z_t}\Big) \nonumber - 2\epsilon Z_t \Big(\frac{X_t}{Z_t}(1 - (1-d/m)^{Z_t})\Big) \pm O(\epsilon^2)
\end{eqnarray*}
The first term in the RHS follows from the exponential discount in the utility of the agents in the pool. 
The second term in the RHS stands for the new arrivals. The third term stands for the perished agents and the last term stands for the the matched agents. 

We use the $e^{-x} \geq 1 - x$ inequality and  rearrange the equation to get,
\begin{eqnarray*}
\E{X_{t + \epsilon} | X_t, Z_t} \geq m\epsilon + X_t - \epsilon X_t (\delta + 2 ) - O(\eps^2).
\end{eqnarray*}
Taking expectation from both sides of the above inequality we get,
$$ \frac{\E{X_{t+\eps}} -\E{X_t}}{\eps} \geq m - (\delta+2)\E{X_t} - O(\eps) $$
Letting $\eps\to 0$, and solving the above differential equation for $t_1$ we get
$$ \E{X_{t_1}} \geq  \frac{m}{\delta+2} (1-e^{-(\delta+2)t_1}).$$
we used the fact that $\E{X_0}=0$.
\end{proof}

By the above lemma,
\begin{eqnarray*}
\W{\algcritical} &\geq& \frac{2}{mT} \int_{t=\tau_{\mix}(\eps)}^T \frac{m}{\delta+2}(1-e^{-(\delta+2)t}) \Big(1-e^{-d/2}(1+\frac{4\log^2 m}{\sqrt{m}})\Big)dt-\frac{3}{m} \\
&\geq & \frac{2(T- \tau_{\mix}(\eps))}{T(\delta+2)}(1-e^{-d/2}(1+\tilde{O}(1/\sqrt{m})) - \frac{2}{T(\delta+2)^2} -\frac{3}{m}
\end{eqnarray*}

\subsection{Welfare of the Greedy Algorithm}
Here, we upper-bound the welfare of the optimum online algorithm, $\opton$, and that immediately upper-bounds the welfare of the Greedy algorithm.
Recall that by \autoref{prop:opton}, for any $T>$, $1/(2d+1+d^2/m)$ fraction of the agents perish 
in $\opton$. On the other hand, by definition of utility, we receive a utility at most 1 from any matched agent. Therefore, even if all of the matched agents receive a utility of 1, (for any $\delta \geq 0)$ 
$$ \W{\algGreedy} \leq \W{\opton} \leq 1-\frac{1}{2d+1+d^2/m}.$$
%
%


\section{Incentive-Compatible Mechanisms}
\label{mechanism}
In this section we design  a dynamic mechanism to elicit the departure times of agents.
As alluded to in \autoref{subsec:contmechanism}, we assume that agents only have statistical knowledge about the market: That is, each agent knows the market parameters,  m,d,1, and the details of the  dynamic mechanism that the market-maker is executing. But no agent knows the actual instantiation of the probability space. 


Each agent $a$  chooses a {\em mixed strategy}, that is she reports getting critical at an infinitesimal time $[t,t+dt]$ with rate $c_a(t)dt$. 
In other words,  each agent $a$ has a clock that ticks with rate $c_a(t)$ at time $t$ and she reports criticality when the clock ticks. 
We assume each agent's strategy function, $c_a(.)$ is {\em well-behaved}, i.e., it is non-negative, continuously  differentiable and continuously integrable. Note that since the agent can only observe the parameters of the market $c_a(.)$ can depend on any parameter in our model but this function is constant in different sample paths of the stochastic process.

A {\em strategy profile} $\C$ is a vector of well-behaved functions for each agent in the market, that is, $\C =[c_a]_{a\in A}$.
For an agent $a$ and a strategy profile $\C$, let $\E{u_\C(a)}$ be the expected utility of $a$ under the strategy profile $\C$. Note that for any $\C,a$, $0\leq \E{u_\C(a)} \leq 1$.
Given a strategy profile $\C=[c_a]_{a\in A}$, let $\C-c_a+\tilde{c}_a$ denote a strategy profile same as $\C$ but for agent $a$ who is playing $\tilde{c}_a$ rather than $c_a$. The following definition introduces our solution concept.
\begin{definition}\label{def:epsnash}
A strategy profile $\C$ is an {\em$\eps$-mixed strategy Nash equilibrium} if for any agent $a$ and any well-behaved function $\tilde{c}_a(.)$, 
$$ 1-\E{u_\C(a)} \leq (1+\eps)(1-\E{u_{\C-c_a+\tilde{c}_a}}). $$
\end{definition}

Note that the solution concept we are introducing here is different from the usual definitions of an $\eps$-Nash equilibrium, where the condition is either $\E{u_\C(a)} \geq \E{u_{\C-c_a+\tilde{c}_a}} - \eps$, or $\E{u_\C(a)} \geq (1-\eps)\E{u_{\C-c_a+\tilde{c}_a}}$. The reason that we are using $1-\E{u_\C(a)}$ as a measure of distance is because we know that under $\apatient{\alpha}$ algorithm, $\E{u_\C(a)}$ is very close to $1$, and so $1-\E{u_\C(a)}$ is a ``lower-order term''. Therefore, by this definition, we are restricting ourself to a stronger equilibrium concept, which requires us to show that in equilibrium agents cannot increase neither their utilities, \emph{nor} the lower-order terms associates with their utilities by a factor more than $\eps$.

Throughout this section let $k^*\in [m/2-2,m-1]$ be the root of \eqref{eq:rootcritical} as defined in \autoref{prop:criticalconcentration}, and let $\beta:=(1-d/m)^{k^*}$.
In this section we show that if $\delta$ (the discount factor) is no more than $\beta$, then the strategy vectors $c_a(t)=0$ for all agents $a$ and $t$ is a $\eps$-mixed strategy Nash equilibrium for $\eps$ very close to zero. In other words, if all other agents are truthful, an agent's utility from being truthful is almost as large as any other strategy. 

\begin{theorem}
\label{thm:mechanismdesign}
 If the market is at stationary and $\delta \leq \beta$, then $c_a(t)=0$ for all $a,t$ is an $O(d^{4}\log^{3}(m)/\sqrt{m})$-mixed strategy Nash equilibrium for $\amechpatient{\infty}$.
\end{theorem}
\medskip


By our market equivalence result (\autoref{lem:equivalence}), \autoref{thm:mechanismdesign} leads to the following corollary.

\begin{corollary}
\label{cor:mechanismdesign}
Let $\bar\alpha = 1/\alpha + 1$ and $\beta(\alpha) = \bar\alpha (1-d/m)^{m/\bar\alpha}$. If the market is at stationary and $\delta \leq \beta(\alpha)$, then $c_a(t)=0$ for all $a,t$ is an $O((d/\bar\alpha)^{4}\log^{3}(m/\bar\alpha)/\sqrt{m/\bar\alpha})$-mixed strategy Nash equilibrium for $\amechpatient{\alpha}$.
\end{corollary}

The proof of the above theorem is involved but the basic idea is very easy.
If an agent reports getting critical at the time of arrival she will receive a utility of $1-\beta$.
On the other hand, if she is truthful (assuming $\delta=0$) she will receive about $1-\beta/2$.
In the course of the proof we show that by choosing any strategy vector $c(.)$ the expected utility of
an agent interpolates between these two number, so it is maximized when she is truthful.


\begin{lemma}
\label{lem:utilitystrategyupperbound}
Let $Z_0$ be at stationary  and suppose $a$ enters the market at time 0.
 If $\delta  < \beta$, and $10 d^{4}\log^{3}(m) \leq \sqrt{m}$, then for any well-behaved function $c(.)$,
$$ \E{u_c(a)} \leq O\Big(d^{4}\log^{3}(m)/\sqrt{m}\Big)\beta +  \frac{2(1-\beta)}{2-\beta+\delta},
$$

\end{lemma}

We prove the lemma by  first  to writing a closed form expression for the utility of $a$ and
then upper-bounding that expression. We start by studying the expected gain of $a$ in a tiny interval $[t,t+\eps]$.

In the following claim we study the probability $a$ is matched in the interval $[t,t+\eps]$ and the probability that it leaves the market in that interval.
\begin{claim}
\label{cl:mechanismmatcht}
For any time $t\geq 0$, and $\eps>0$,
\begin{align}
\P{a \in M_{t,t+\eps}} &= \eps\cdot \P{a\in A_t} (2+c(t))\E{1-(1-d/m)^{Z_t}-1| a\in A_t} \pm O(\eps^2)\\
\P{a\notin A_{t+\eps},a\in A_t} &=  \P{a\in A_t}  (1-\eps(1 + c(t) + \E{1-(1-d/m)^{Z_t-1} | a\in A_t})  \pm O(\eps^2))
\end{align}
\end{claim}
\begin{proof}
The claim follows from two simple observations. First, $a$ becomes critical in the interval $[t,t+\eps]$
with probability $\eps\cdot \P{a\in A_t} (1+c(t))$ and if he is critical he is matched with probability $\E{(1-(1-d/m)^{Z_t-1} | a\in A_t}$. Second, $a$ may also get matched (without getting critical) in the interval $[t,t+\eps]$. Observe that  if an agent $b\in A_t$ where $b\neq a$ gets critical she will be matched with $a$ with probability $ 1-(1-d/m)^{Z_t-1} / (Z_t-1)$,.
Therefore, the probability that $a$ is matched at $[t,t+\eps]$ without getting critical is
\begin{align*} \P{a\in A_t} &\cdot \E{\eps\cdot (Z_t-1)\frac{1-(1-d/m)^{Z_t-1}}{Z_t-1} | a\in A_t} \\
&= \eps\cdot \P{a\in A_t} \E{1-(1-d/m)^{Z_t-1} | a\in A_t}
\end{align*}
The claim follows from simple algebraic manipulations. 
\end{proof}

We need to  study the conditional expectation $\E{1-(1-d/m)^{Z_t-1} | a\in A_t}$ to use the above claim. This is not easy in general; although the distribution of $Z_t$ remains  stationary, the distribution of $Z_t$ conditioned on $a\in A_t$ can be a very different distribution. So, here we prove simple upper and lower bounds on $\E{1-(1-d/m)^{Z_t-1} | a\in A_t}$ using the concentration properties of $Z_t$.
By the assumption of the lemma $Z_t$ is at stationary at any time $t\geq 0$. 
Let $k^*$ be the number defined in \autoref{prop:criticalconcentration},  and $\beta=(1-d/m)^{k^*}$.
Let $\sigma:=\sqrt{6\log(8m/\beta)}$. By \autoref{prop:criticalconcentration}, for any $t\geq 0$,
\begin{align}
 \E{1-(1-d/m)^{Z_t-1} | a\in A_t} &\leq \E{1-(1-d/m)^{Z_t-1} | Z_t < k^*+\sigma\sqrt{4m},a\in A_t}  \nonumber \\
&~~+\P{Z_t\geq k^*+\sigma\sqrt{4m} | a\in A_t}\nonumber \\
&\leq 1-(1-d/m)^{k^*+\sigma\sqrt{4m}}+ \frac{\P{Z_t \geq k^*+\sigma\sqrt{4m}}}{\P{a\in A_t}}  \nonumber \\
& \leq 1-\beta +  \beta(1-(1-d/m)^{\sigma\sqrt{4m}})+ \frac{8\sqrt{m} e^{-\sigma^2/3}}{\P{a\in A_t}}\nonumber \\
&\leq 1-\beta+\frac{2 \sigma d \beta }{\sqrt{m}} + \frac{\beta}{m^2 \cdot \P{a\in A_t}}
 \label{eq:mechanismmatchupper}
 \end{align}
  In the last inequality we used \eqref{eq:bernoulli} and the definition of $\sigma$.
 Similarly,
 \begin{align}
 \E{1-(1-d/m)^{Z_t-1} | a\in A_t}  &\geq \E{1-(1-d/m)^{Z_t-1} | Z_t \geq k^*-\sigma\sqrt{4m},a\in A_t}\nonumber \\
 &~~\cdot \P{Z_t \geq k^*-\sigma\sqrt{4m} | a\in A_t}\nonumber\\
 &\geq (1-(1-d/m)^{k^*-\sigma\sqrt{4m}}) \frac{\P{a\in A_t} - \P{Z_t < k^*-\sigma\sqrt{4m}}}{\P{a\in A_t}}\nonumber \\
&\geq 1-\beta -\beta((1-d/m)^{-\sigma\sqrt{4m}}-1) -\frac{2\sqrt{m}e^{-\sigma^2}}{\P{a\in A_t}}\nonumber \\
&\geq  1-\beta -  \frac{4d\sigma \beta}{\sqrt{m}} -\frac{\beta^3}{m^3 \cdot \P{a\in A_t}}
\label{eq:mechanismmatchlower}
\end{align}
where in the last inequality we used \eqref{eq:bernoulli}, the assumption that $2d\sigma\leq \sqrt{m}$
and the definition of $\sigma$.

Next, we write a closed form upper-bound for $\P{a\in A_t}$. 
Choose $t^*$ such that $\int_{t=0}^{t^*} (2+c(t)) dt = 2\log(m/\beta)$. 
Observe that $t^* \leq \log(m/\beta) \leq \sigma^2/6$. Since $a$ leaves the market with rate at least $1+c(t)$ and at most $2+c(t)$,
we can write  
\begin{equation}
\label{eq:tstar}
\frac{\beta^2}{m^2} = \exp\Big(-\int_{t=0}^{t^*} (2+c(t))dt\Big) \leq \P{a\in A_{t^*}} \leq \exp\Big(-\int_{t=0}^{t^*} (1+c(t))dt\Big)\leq  \frac{\beta}{m}
\end{equation}  
Intuitively, $t^*$ is a moment where the expected utility of that $a$ receives in the interval $[t^*,\infty)$ is negligible, i.e., in the best case it is at most $\beta/m$. 

By  \autoref{cl:mechanismmatcht} and \eqref{eq:mechanismmatchlower}, for any $t\leq t^*$,
\begin{eqnarray*} 
\frac{\P{a\in A_{t+\eps}} - \P{a\in A_t}}{\eps}
& \leq& -\P{a\in A_t}\Big(2 + c(t) - \beta-\frac{4d\sigma\beta}{\sqrt{m}} -\frac{\beta^3}{m^3 \cdot \P{a\in A_t}} \pm O(\eps)\Big)\nonumber\\
&\leq &-\P{a\in A_t}\Big(2+c(t)-\beta-\frac{5d\sigma\beta}{\sqrt{m}}\pm O(\eps)\Big)
\end{eqnarray*}
where in the last inequality we used \eqref{eq:tstar}.
Letting $\eps\to 0$, for $t\leq t^*$, the above differential equation yields,
\begin{equation}
\label{eq:ptupper}
 \P{a\in A_t} \leq \exp\Big(-\int_{\tau=0}^t \Big(2+c(\tau)-\beta -\frac{5d\sigma\beta}{\sqrt{m}}\Big) d\tau\Big) \leq \exp\Big(-\int_{\tau=0}^t (2+c(\tau)-\beta)d\tau\Big) + \frac{2 d\sigma^3\beta}{\sqrt{m}}.
 \end{equation}
where in the last inequality we used $t^*\leq \sigma^2/6$, $e^x\leq 1+2x$ for $x\leq 1$ and lemma's assumption $5 d\sigma^2\leq \sqrt{m}$ . 

Now, we are ready to upper-bound the utility of $a$.
By \eqref{eq:tstar} the expected utility that $a$ gains after $t^*$ is no more than $\beta/m$.
Therefore,
\begin{eqnarray*} \E{u_c(a)} &\leq & \frac{\beta}{m} +\int_{t=0}^{t^*} (2+c(t)) \E{1-(1-d/m)^{Z_t-1}|a\in A_t} \P{a\in A_t} e^{-\delta t}dt \\
&\leq & \frac{\beta}{m} +\int_{t=0}^{t^*} (2+c(t)) ((1-\beta)\P{a\in A_t}+\beta/\sqrt{m}) e^{-\delta t}dt\\
&\leq & \frac{\beta}{m} + \int_{t=0}^{t^*} (2+c(t)) \Big((1-\beta) \exp\Big(-\int_{\tau=0}^t (2+c(\tau)-\beta)d\tau \Big) + \frac{3 d\sigma^3}{\sqrt{m}}\beta\Big) e^{-\delta t}dt \\
&\leq& \frac{2d \sigma^5}{\sqrt{m}}\beta + \int_{t=0}^{\infty} (1-\beta) (2+c(t)) \exp\Big(-\int_{\tau=0}^t (2+c(\tau)-\beta)d\tau\Big) e^{-\delta t}dt.
\end{eqnarray*}
 In the first inequality we used equation \eqref{eq:mechanismmatchupper}, 
 in second inequality we used equation \eqref{eq:ptupper}, and in the last 
 inequality we use the definition of $t^*$. 
 We have finally obtained a closed form upper-bound on the expected utility of $a$.

Let $U_c(a)$ be the right hand side of the above equation.
Next, we show that $U_c(a)$ is maximized by letting $c(t)=0$ for all $t$.
This will complete the proof of \autoref{lem:utilitystrategyupperbound}.
Let $c$ be a function that maximizes $U_c(a)$ which is not equal to zero. Suppose $c(t)\neq 0$ for some $t\geq 0$. We define a function $\tilde{c}:\R_+\to\R_+$ and we show that if $\delta < \beta$, then $U_{\tilde{c}}(a) > U_c(a)$. 
Let $\tilde{c}$ be the following function,
$$ \tilde{c}(\tau)=\begin{cases}
c(\tau) & \text{if $\tau<t$},\\
0 & \text{if $t\leq \tau\leq t+\eps$,}\\
c(\tau)+c(\tau-\eps) & \text{if $t+\eps\leq\tau\leq t+2\eps$,}\\
c(\tau) & \text{otherwise}.
\end{cases}$$
In words, we push the mass of $c(.)$ in the interval $[t,t+\eps]$ to the right. We remark that the above function $\tilde{c}(.)$ is not necessarily continuous so we need to smooth it out. The latter  can be done without introducing any errors and we do not describe the details here. Let $S:=\int_{\tau=0}^t (1+c(t)+\beta)d\tau$. Assuming
$\tilde{c}'(t) \ll 1/\eps$, we have
\begin{eqnarray*}
U_{\tilde{c}}(a)-U_c(a) &\geq & -\eps \cdot c(t) (1-\beta) e^{-S} e^{-\delta t} + \eps\cdot c(t) (1-\beta)  e^{-S-\eps(2-\beta)} e^{-\delta(t+\eps)} \\
&&+ \eps (1-\beta)(2+c(t+\eps)) (e^{-S-\eps(2-\beta)} e^{-\delta (t+\eps)}-e^{-S-\eps(2+c(t)-\beta)}e^{-\delta(t+\eps)}) \\
&=& -\eps^2 \cdot c(t)(1-\beta)e^{-S-\delta t}  (2-\beta+\delta) + \eps^2(1-\beta)(2+c(t+\eps)) e^{-S-\delta t} c(t)\\
&\geq & \eps^2\cdot(1-\beta) e^{-S-\delta t} c(t) (\beta-\delta).
\end{eqnarray*}
Since $\delta < \beta$ by the lemma's assumption, the maximizer of $U_c(a)$ is the all zero function. 
Therefore, for any well-behaved function $c(.)$,
\begin{eqnarray*} \E{u_c(a)} &\leq& \frac{2d\sigma^5}{\sqrt{m}}\beta+ \int_{t=0}^\infty 2(1-\beta)\exp\Big(-\int_{\tau=0}^t (2-\beta)d\tau\Big) e^{-\delta t}dt \\
&\leq & O(\frac{ d^{4}\log^{3}(m)}{\sqrt{m}})\beta + \frac{2(1-\beta)}{2-\beta+\delta}.
\end{eqnarray*}
In the  last inequality we used that $\sigma=O(\sqrt{\log(m/\beta)})$ and $\beta\leq e^{-d}$.
This completes the proof of \autoref{lem:utilitystrategyupperbound}.

The proof of \autoref{thm:mechanismdesign} follows simply from the above analysis.
\begin{proofof}{\autoref{thm:mechanismdesign}}
All we need to do is to lower-bound the expected utility of an agent $a$ if she is truthful.
We omit the details as they are essentially similar. So, if all agents are truthful, 
$$\E{u(a)} \geq \frac{2(1-\beta)}{2-\beta+\delta} - O\Big(\frac{d^{4}\log^{3}(m)}{\sqrt{m}}\Big) \beta.$$
This shows that the strategy vector corresponding to truthful agents is an $O(d^{4}\log^{3}(m)/\sqrt{m})$-mixed  strategy Nash equilibrium. 
\end{proofof}

\section{Concluding Remarks}\label{conclusion}

In this paper, we developed a new framework to model dynamic matching markets.  This paper innovates by accounting for stochastic departures and analyzing the problem under a variety of information conditions.  Rather than modeling market thickness via a fixed match-function, it explicitly accounts for the network structure that affects the Planner's options.  This allows market thickness to emerge as an endogenous phenomenon, responsive to the underlying constraints. 

There are many real market design problems where the timing of transactions must be decided by a policymaker.  These include paired kidney exchanges, dating agencies, and online labor markets such as oDesk.  In such markets, policymakers face a trade-off between the speed of transactions and the thickness of the market.  It is natural to ask, ``Does it matter \emph{when} transactions occur?  How much does it matter?''  The first insight of this paper is that waiting to thicken the market can yield substantial welfare gains, and this result is quite robust to the presence of waiting costs.

The second insight of this paper relates to optimization.  Because our approach takes seriously the network structure of the planner's constraints, the resulting Markov decision process is combinatorially complex, and not tractable via standard dynamic programming methods.  Surprisingly, we find that na\"{i}ve local algorithms with different waiting properties can come close to optimal benchmarks that exploit the whole graph structure.

The third insight of this paper is that information and waiting time are complements; even short-horizon information about agents' departure times yields large gains that can be exploited by simple waiting algorithms.  When the urgency of individual cases is private information, we exhibit a mechanism without transfers that elicits such information from sufficiently patient agents.

These results suggest that the dimension of time is a first-order concern in many matching markets, with welfare implications that static models do not capture.  They also suggest that policymakers would reap large gains from acquiring timing information about agent departures, such as by paying for predictive diagnostic testing or monitoring agents' outside options.  

A key technical contribution of this paper is that we show how to characterize dynamic matching markets as analytically tractable Markov processes.  We develop new techniques to prove concentration bounds on stochastic processes; in particular, to show that for a large $t$, a given function is concentrated in an interval whose size does not depend on $t$.  These enable us to deliver analytic comparisons of different algorithms, even when their stationary distributions do not have simple closed-form characterizations.  We hope these techniques are helpful to other researchers interested in dynamic matching problems.

We suggest some promising extensions.  First, one could generalize the model to allow multiple \emph{types} of agents, with the probability of an acceptable transaction differing across type-pairs.  This could capture settings where certain agents are known \emph{ex ante} to be less likely to have acceptable trades than other agents, as is the case for patients with high Panel Reactive Antibody (PRA) scores in paired kidney exchanges.  The multiple-types framework also contains bipartite markets as a special case.

Second,  one could adopt more gradual assumptions about agent departure processes; agents could have a range of individual \emph{states} with state-dependent perishing rates, and an independent process specifying transition rates between states.  Our model, in which agents transition to a critical state at rate $\lambda$ and then perish imminently, is a limit case of the multiple-states framework.

Third, it would be interesting to enrich the space of preferences in the model, such as by allowing matches to yield a range of payoffs.  We conjecture that this would reinforce our existing results, since waiting to thicken the market could allow planners to make better matches, in addition to increasing the size of the matching.  Further insights may come by making explicit the role of price in dynamic matching markets.  It is not obvious how to do so, but similar extensions have been made for static models \cite{KC82,HM05}, and the correct formulation may seem obvious in retrospect.

Much remains to be done in the theory of dynamic matching.  As market design expands its reach, re-engineering markets from the ground up, economists will increasingly have to answer questions about the timing and frequency of transactions.  Many dynamic matching markets have important features (outlined above) that we have not modeled explicitly.  We offer this paper as a step towards systematically understanding matching problems that take place across time.

\appendix

\bibliographystyle{alpha}
\bibliography{references}

\appendix
\section{Auxiliary Inequalities}
In this section we prove several inequalities that are used throughout the paper.
For any $a,b \geq 0$, 
\begin{eqnarray}
\sum_{i=a}^\infty e^{-bi^2} = \sum_{i=0}^\infty e^{-b(i+a)^2}  \leq \sum_{i=0}^\infty e^{-ba^2 -2iab}
&=& e^{-ba^2} \sum_{i=0}^\infty (e^{-2ab})^i \nonumber \\
&=& \frac{e^{-ba^2}}{1-e^{-2ab}} \leq \frac{e^{-ba^2}}{\min\{ab,1/2\}}.
\label{eq:sumsqexp}
\end{eqnarray}
The last inequality can be proved as follows: If $2ab\leq 1$, then $e^{-2ab} \leq ab$, otherwise
$e^{-2ab} \leq 1/2$.

For any $a,b\geq 0$,
\begin{equation}
\label{eq:sumexpexp}
\sum_{i=a}^\infty (i-1)e^{-bi^2}  \leq \int_{a-1}^\infty xe^{-bx^2}dx  = \frac{-1}{2b} e^{-bx^2}\mid^\infty_{a-1} = \frac{e^{-b(a-1)^2}}{2b}.
\end{equation}

For any $a \geq 0$ and $0\leq b\leq 1$,
\begin{eqnarray}
\label{eq:sumexp}
\sum_{i=a}^\infty ie^{-bi} = e^{-ba} \sum_{i=0}^\infty (i+a)e^{-bi} = e^{-ba} \Big(\frac{a}{1-e^{-b}} + \frac{1}{(1-e^{-b})}\Big) \leq \frac{e^{-ba} (2ba+4)}{b^2}.
\end{eqnarray}

The Bernoulli inequality states that for any $x\leq 1$, and any  $n\geq 1$,
\begin{equation}
\label{eq:bernoulli}
 (1-x)^n \geq 1-xn.
 \end{equation}
Here, we prove for integer $n$. The above equation can be proved by a simple induction on $n$. It trivially holds for $n=0$. Assuming it holds for $n$ we can write,
 $$ (1-x)^{n+1} = (1-x)(1-x)^n \geq (1-x) (1-xn) = 1-x(n+1)+x^2n \geq 1-x(n+1).$$

\end{document}